\documentclass[aps,pra,showpacs,twoside,twocolumn,10pt]{revtex4-2}
\usepackage[colorlinks=true, citecolor=blue, urlcolor=blue, linkcolor = blue ]{hyperref}

\usepackage{amsthm,graphicx,bm,amsmath,xcolor,braket,plain,color,amsthm,amsfonts,tabularx,graphicx,bbm,mathtools,esvect,wrapfig,verbatim,enumitem,dcolumn,amssymb,appendix,physics,fmtcount,booktabs,csquotes,epsfig,times,geometry,array,mathrsfs,graphics,epstopdf,latexsym,comment,cancel}

\usepackage[normalem]{ulem}

\usepackage[mathscr]{euscript}
\def\Tr{\operatorname{Tr}}

\newtheorem{theorem}{Theorem}

\begin{document}

\title{Encoding parameters by measurement: Forgetting can be better in quantum metrology}

\author{Shuva Mondal}
\email{mondalshuva01@gmail.com}

\author{Priya Ghosh}
\email{priyaghosh1155@gmail.com}

\author{Ujjwal Sen}
\email{ujjwal@hri.res.in}

\affiliation{Harish-Chandra Research Institute, A CI of Homi Bhabha National Institute, Chhatnag Road, Jhunsi, Prayagraj 211 019, India}

\begin{abstract}
We introduce quantum parameter estimation with the encoding being via a quantum measurement. We quantify the precision for estimating parameters characterizing a general two-outcome qubit measurement, considering two cases: when the outcomes of the \textit{encoding} measurement are recorded and when the same are ignored. We find that in a large variety of  
such estimation scenarios, forgetting the outcomes yields higher precision. We derive a necessary criterion under which remembering the measurement outcomes provides better precision in comparison to the outcome-forgotten strategy. Furthermore, we establish a necessary and sufficient criterion for the simultaneous estimation of multiple parameters encoded by an arbitrary quantum process, including those involving measurements, using qubit probes, and find when the quantum Cram\'er–Rao bound 
is valid and achievable. For simultaneous estimation of two parameters characterizing the measurement, we find that the achievable quantum Cram\'er-Rao bound can be a valid precision bound only when the measurement direction depends on the parameters of interest.
\end{abstract}

\maketitle

\section{Introduction}
Several quantum information protocols rely on performing specific measurements to unlock quantum advantages. For example, perfect quantum teleportation~\cite{Bennett1993,Pirandola2015} using a singlet requires a Bell measurement on Alice’s qubits, while the maximal violation of Bell inequalities~\cite{bell_aspect_2004,bell-review} in the (2,2,2) Bell scenario using a singlet state is achieved only under a precise choice of local measurements on both sides. These tasks therefore require the noiseless specific measurements. In realistic quantum devices, however, noise is unavoidable and can introduce small, unknown deviations in the implemented measurements. This motivates the study of quantum estimation of parameters encoded by a measurement,
which is fundamentally different from known encoding schemes such as unitary or channel parameter estimation.

Quantum metrology~\cite{Helstrom1969,HOLEVO1973337,PhysRevLett.72.3439,doi:10.1126/science.1104149,PhysRevLett.94.020502,PhysRevLett.96.010401,Giovannetti2011,holevo,RevModPhys.89.035002,RevModPhys.90.035005,RevModPhys.90.035006,Pirandola2018,mondal2024multicriticalquantumsensorsdriven,agarwal2025criticalquantummetrologyusing,agarwal2025quantumsensingultracoldsimulators} is the area of study in which one uses quantum mechanical principles to estimate a single parameter or a set of parameters, as precisely as possible. Quantum parameter estimation has been widely explored when the parameters are encoded by unitary processes~\cite{PhysRevA.110.012620,Paris2024,bhattacharyya2025quantumsensingevenversus,bhattacharyya2025precisionestimatingindependentlocal,saha2025entanglementconstrainedquantummetrologyrapid,pal2025rolephaseoptimalprobe,chaki2025nonpositivemeasurementsarentbeneficial}, quantum channels~\cite{PhysRevA.63.042304, Akio_Fujiwara_2003, MasahitoHayashi_2010,demkowicz_2012,rafal_2014, Pirandola17, RAZAVIAN2019825, PRXQuantum.2.010343,rafal_2023,mondal2025optimalquantumprecisionnoise}, unitary encoding under noise~\cite{Huelga97,Alipur14,Yousefjani17,Peng23,Bhattacharyya24a}, and related scenarios~\cite{Mehboudi2019-thermometry,Agrawal2025indefinitetime,sarkar2025}.
In this work, we study the estimation of parameters characterizing arbitrary two-outcome qubit measurements. We consider two distinct strategies: (i) in the first, the outcome of the encoding measurement is recorded; the post-measurement state corresponding to each outcome is treated as the encoded state, and the estimation error is computed for each measurement outcome, weighted by its occurrence probability, leading to a probabilistic average estimation error, referred to as the “outcome-remembered error (OR error)”; (ii) in the second, the measurement outcome is not accessible; the encoded state is the non-selective post-measurement state obtained when one forgets the measurement outcomes, and we analyze the estimation error for this particular encoded state, referred to as the “outcome-forgotten error (OF error).” We compare these two estimation errors and address the key question: does remembering measurement outcomes enhance metrological precision in measurement-parameter estimation?

We focus on the estimation of parameters encoded by a general two-outcome qubit measurement.
We find strikingly that the OF error yields strictly better precision than the OR error, a behavior that persists even after optimizing over all probes in a large variety of such estimation scenarios.
More specifically, we establish a necessary criterion under which remembering the measurement outcomes can provide better precision than the outcome-forgotten strategy in estimating any parameter of the encoding measurement.

We then turn to the simultaneous estimation of multiple parameters and derive a necessary and sufficient criterion in the estimation of multiple parameters encoded by an arbitrary encoding process using qubit probes, and find 
when the quantum Fisher information matrix (QFIM), associated with the achievable  quantum Cram\'er–Rao bound (QCRB), is invertible. It has been shown that when any two unitarily-encoded parameters are estimated simultaneously with qubit probes, imposing the achievability of the QCRB constraint always renders the QFIM singular~\cite{Paris2024}, implying that the QCRB ceases to be a valid precision bound in such scenarios. 
In particular, in our work, we derive a necessary and sufficient criterion for the invertibility of QFIM associated with an achievable QCRB
for arbitrary two-parameter estimation scenarios, where the parameters can be encoded through any encoding process. However, this situation changes fundamentally in multiple parameter estimation settings: for any multiparameter encoding scheme involving more than two simultaneously estimated parameters, the QFIM associated with an achievable QCRB is necessarily singular.
Furthermore, we find that in parameter-encoding schemes implemented through measurement, if assessment of measurement direction is not a part of simultaneous estimation of two parameters, the QFIM associated with the achievable QCRB is always singular in all strategies. However, when estimating direction is included in the multiparameter metrology, QFIM singularity is not necessary, irrespective of whether the measurement outcomes are remembered or forgotten.

The rest of the paper is organized as follows. In Sec.~\ref{pre}, we briefly discuss a few well-known concepts related to this work: quantum metrology and quantum measurement. 
In Sec.~\ref{subsection:problem_statement}, we quantify the estimation errors for the parameters encoded by the measurement in two scenarios: (i) when the outcome of the encoding measurement is recorded, and (ii) when the outcome is ignored. 
We present our result obtained in the estimation of a single parameter characterizing an arbitrary two-outcome qubit measurement in Sec.~\ref{single parameter}, and find a necessary criterion under which remembering individual measurement outcomes can provide better precision in the single measurement-parameter estimation scenario over the outcome-forgotten strategy.
In Sec.~\ref{multiple-parameters-general}, we show that for any encoding, QFIM, associated to achievable QCRB, is always singular if number of parameters to be estimated simultaneously is greater than two, followed by a necessary and sufficient condition when the number of parameters is exactly two, while in Sec.~\ref{multiple-parameters-measurement},
we analyze the scenario of multiparameter estimation of parameters encoded by measurement and
further show that, unlike the unitary encoding two-parameter estimation using qubit probes, the QFIM associated with the achievable QCRB does not necessarily become singular when two measurement parameters are estimated simultaneously.
Finally, we conclude in Sec.~\ref{conclusion}.

\section{Setting the stage}
\label{pre}
Here, we discuss two basic concepts relevant to this work: quantum metrology and positive operator-valued measurements (POVMs).

The goal of quantum metrology~\cite{Helstrom1969,HOLEVO1973337,PhysRevLett.72.3439,Liu_2020} is to estimate unknown parameters as precisely as possible using quantum resources, based on the fundamental principles of quantum mechanics. 
In quantum multiparameter estimation theory, the estimation error, quantified by the covariance matrix of an estimator and denoted by $\mathbf{V}(\bm{x})$, is bounded by the inverse of the quantum Fisher information matrix (QFIM), known as the quantum Cram\'er–Rao bound (QCRB), i.e.,
\begin{align}
\label{eq:multiparameter-quantum_cramer_rao}
    \mathbf{V}(\bm{x})\geq \mathbb{\mathbf{F}}_\mathbf{Q}^{-1}\left(
    \Lambda_{\bm{x}}(\rho_{\text{in}})\right).
\end{align}
$\mathbb{\mathbf{F}}_\mathbf{Q}\left(\Lambda_{\bm{x}}(\rho_{\text{in}})\right)$ 
denotes the QFIM of the encoded state $\rho(\bm{x})\coloneqq \Lambda_{\bm{x}}(\rho_{\text{in}})$, where $\Lambda_{\bm{x}}(\bullet)$ is the quantum process that encodes $m\in\mathbb{N}$ parameters $\bm{x}=\{x_1,\dots,x_m\}$, the parameters intended to be estimated, into the input state $\rho_{\text{in}}$.
The elements of QFIM are defined as
\begin{align}
\label{eq:QFIM with SLD}
[\mathbb{\mathbf{F}}_\mathbf{Q}(\rho(\bm{x}))]_{ij}=\frac{1}{2}\Tr\left[\rho(\bm{x})\{L_{x_i}L_{x_j}+L_{x_i}L_{x_j}\}\right],
\end{align}
where symmetric logarithmic derivative (SLD) operator $L_{x_i}$ corresponding to the parameter $x_i$
satisfies $\frac{\partial \rho(\bm{x})}{\partial x_i}=\frac{1}{2}\left[L_{x_i} \rho(\bm{x}) + \rho(\bm{x}) L_{x_i}\right].$ For single-parameter estimation, the QCRB is always achievable and valid.
In the multiparameter setting, however, it faces two issues: (i) it may be unachievable due to incompatibility of the optimal measurements associated with different parameters, diagnosed via the Uhlmann matrix~\cite{Carollo2018}
\begin{align}
\label{eq:Uhlmann_matrix}
[\mathcal{D}(\rho(\bm{x}))]_{ij} = \frac{1}{2} \Tr \left[\rho(\bm{x})\{L_{x_i}L_{x_j} - L_{x_j}L_{x_i}\}\right],
\end{align}
which must vanish for QCRB achievability, and (ii) the QFIM can be singular.
In such cases, alternative bounds like the Holevo–Cramér–Rao bound~\cite{holevo} (HCRB), which is attainable, nonsingular but tighter
compared to the QCRB, are often used.
Further details are provided in Appendix~\ref{section:quantum_metrology}.

Now we move into the discussion of the quantum measurements. Let us denote a general positive operator-valued measurement (POVM) by $\{E_i\}$. 
If the POVM $\{E_i\}$ is performed on a quantum state $\rho$, the probability that the outcome associated with $E_i$ clicks and the corresponding post-measurement state are given by $Q_i=\Tr[E_i\rho]$ and, $\frac{\sqrt{E_i}\rho\sqrt{E_i}}{Q_i}$ respectively. On the other hand, if the measurement outcomes are forgotten or not recorded, the resulting (non-selective) post-measurement state becomes $\sum_i \sqrt{E_i}\, \rho \, \sqrt{E_i}$.

In the remainder of the paper, we denote a $d$-dimensional Hilbert space by $\mathcal{H}_d$, and the identity operator acting on that space by $\mathbb{I}_d$. $ \vec{\sigma} \coloneqq \{ \sigma_x, \sigma_y, \sigma_z\}$ refer to the vector of the standard Pauli operators. In the single-parameter setting, we denote the parameter of interest by $x_1$, whereas in the multiparameter case, the set of $m$ parameters is written as $\{x_1, x_2, \dots, x_m\}$.

\section{Our pursuit}
\label{subsection:problem_statement}
In this work, we study the estimation of a parameter or a set of parameters that characterize a quantum measurement.
In this section, we study two estimation strategies in such estimation scenarios—one where encoding measurement outcomes are recorded and another where the outcomes are unknown—and quantify the best achievable estimation errors of the parameter of interest in each case.

Suppose an $k$-outcome POVM $\{E_i(x_1)\}$, parameterized by a parameter $x_1$, is performed on an input state $\rho_{\text{in}}\in \mathcal{H}_d$,
where $x_1$ is the measurement parameter of interest that we wish to estimate.
Let $\rho_{f_i}(x_1)$ denote the post-measurement state conditioned on the click of the $i^{\text{th}}$ outcome, obtained with probability
$Q_i(x_1)$.
When the measurement outcomes are recorded, the achievable precision in estimating the parameter $x_1$ can be quantified via the following quantity
\begin{equation}
\label{eq-OR}
\Delta_{\textbf{OR}}  (x_1) \coloneqq \sum_{i=1}^k \frac{Q_i(x_1)}{\sqrt{\mathbb{F}_Q(\rho_{f_i}(x_1))}},
\end{equation}
where $\mathbb{F}_Q(\rho_{f_i}(x_1))$ denotes the quantum Fisher information (QFI) of the encoded state $\rho_{f_i}(x_1)$.
We refer to $\Delta_{\textbf{OR}} (x_1) $ as the outcome-remembered error (OR error, for short).
On the other hand, if the record of measurement outcomes is not kept, the minimum achievable error to estimate $x_1$ is  given by
\begin{equation}
\label{eq-OF}
\Delta_{\textbf{OF}} (x_1) \coloneqq \frac{1}{\sqrt{\mathbb{F}_Q(\rho_F(x_1))}},
\end{equation}
where $\rho_F(x_1)$ is the encoded state—defined as the non-selective post-measurement state obtained by forgetting the outcomes after the measurement $\{E_i(x_1)\}$ on
the input state $\rho_{\text{in}}$.
We name $\Delta_{\textbf{OF}} (x_1) $ as outcome-forgotten error (OF error, for short).

Extending the single measurement parameter estimation case to multiple measurement parameters $\bm{x} \coloneqq \{x_j\}$, the OR error and the OF error will have the following form:
\begin{align*}
\Delta_{\textbf{OR}}  (\bm{x} ) &\coloneqq \sum_{i=1}^k \mathbf{Q}_i(\bm{x}) \sqrt{\Tr[\mathcal{W} \mathbb{\mathbf{F}}_\mathbf{Q}^{-1}(\bm{\rho}_{f_i}(\bm{x})})],\\
\Delta_{\textbf{OF}}  (\bm{x}) &\coloneqq \sqrt{\Tr[\mathcal{W} \mathbb{\mathbf{F}}_\mathbf{Q}^{-1}(\bm{\rho_F}(\bm{x}))]},
\end{align*}
when an $k$-outcome POVM $\{\bm{E}_i(\bm{x})\}$, parameterized by parameters $\bm{x}$, is performed on an input state $\rho_{\text{in}}$. 
Here, $\bm{\rho}_{f_i}(\bm{x})$ is the post-measurement state conditioned on the $i^{\text{th}}$ outcome, obtained with probability $\mathbf{Q}_i(\bm{x})$, while $\bm{\rho_F}(\bm{x})$ is the non-selective post-measurement state when outcomes are not recorded.
$\mathbb{\mathbf{F}}_\mathbf{Q}(\bm{\rho}_{f_i}(\bm{x}))$, and $\mathbb{\mathbf{F}}_\mathbf{Q}(\bm{\rho_F}(\bm{x}))$ denote the QFIM of the encoded state $\bm{\rho}_{f_i}(\bm{x})$ and $\bm{\rho_F}(\bm{x})$ respectively.
The matrix $\mathcal{W}$ is a real, symmetric, and positive-definite cost (weight) matrix introduced to convert the QCRB into a scalar figure of merit for multiparameter estimation.
If the QFIM of the encoded states, $\mathbb{\mathbf{F}}_\mathbf{Q}(\bm{\rho}_{f_i}(\{x_j\}))$ and $\mathbb{\mathbf{F}}_\mathbf{Q}(\bm{\rho_F}(\{x_j\}))$, face any of the issues discussed earlier in Sec.~\ref{pre} (e.g., singularity of QFIM or unattainability of QCRB), alternative attainable bounds can be used to define $\Delta_{\textbf{OR}}$ and $\Delta_{\textbf{OF}}$ in multiple measurement-parameter estimation.
We denote the $\Delta_{\textbf{OR}}$ and  $\Delta_{\textbf{OF}}$, each minimized over all input states, by $\bar{\Delta}_{\textbf{OR}}$ and $\bar{\Delta}_{\textbf{OF}}$ respectively, throughout our work. Also an important comment regarding our notation: we will denote the encoded state as $\rho_F$ when measurement outcomes are forgotten and $\rho_{f_i}$ to denote the encoded state when post-selected states corresponds to clicking of $E_i$.

\begin{figure*}
    \centering \hspace{-12 mm}
    \begin{minipage}[t]{0.47\textwidth}
        \centering
        \includegraphics[trim=0cm 0.0cm 0.0cm 0, clip, height=5.9 cm]{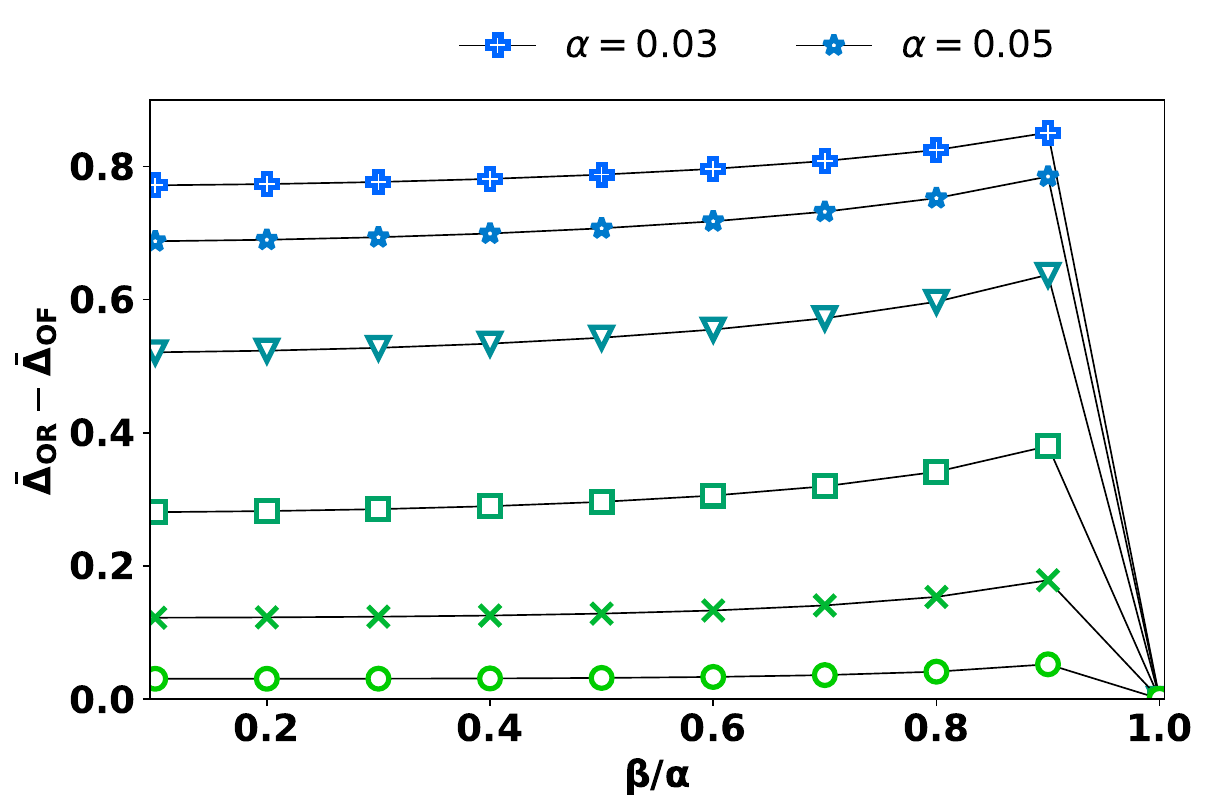}\\
        \textbf{(a)}
    \end{minipage}\hspace{2mm}
    \begin{minipage}[t]{0.47\textwidth}
        \centering
        \includegraphics[trim=0cm 0.0cm 0.0cm 0, clip,  height=5.9 cm]{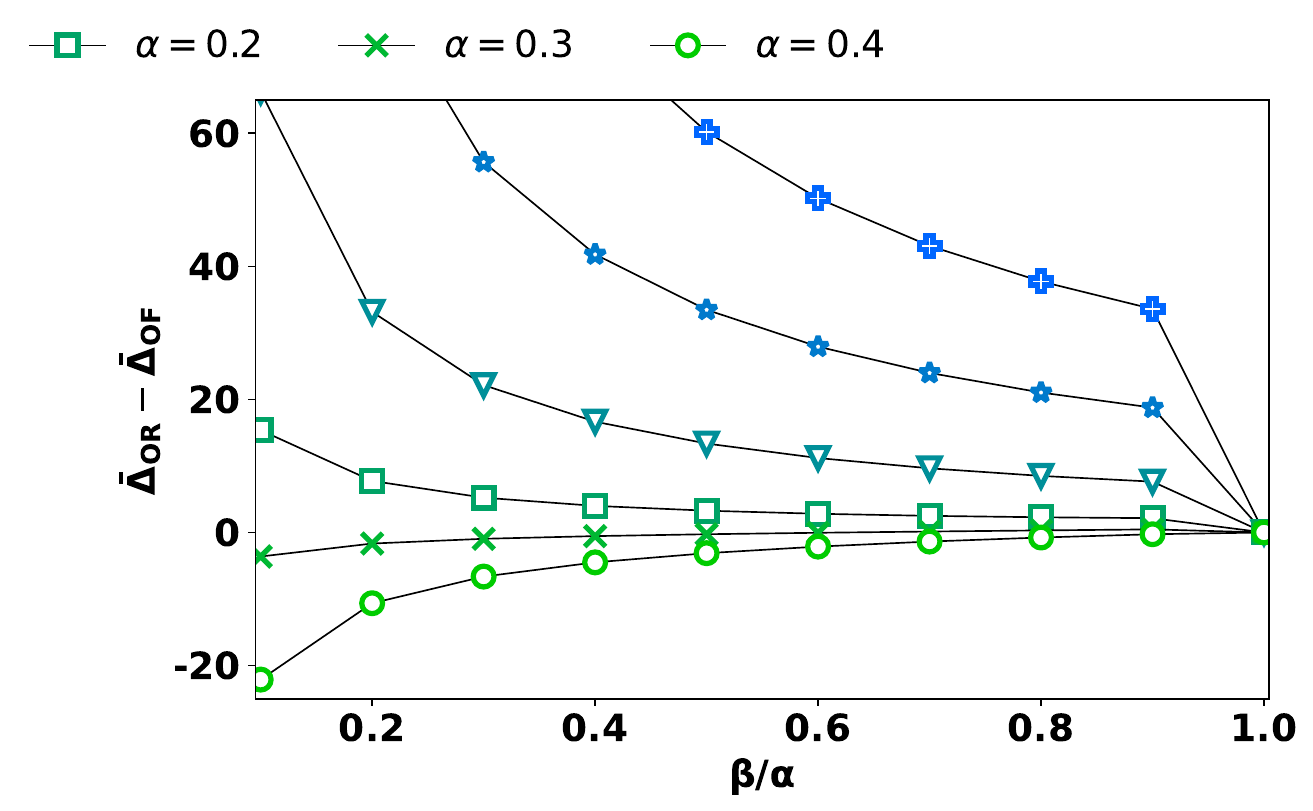}\\
        \textbf{(b)}
    \end{minipage}
    \caption{(a) The left plot shows $(\bar{\Delta}_{\textbf{OR}}-\bar{\Delta}_{\textbf{OF}})$ along the vertical axis while the ratio $\alpha/\beta$ is varied along the horizontal axis for estimation of the parameter of $\beta$. The $(\bar{\Delta}_{\textbf{OR}}-\bar{\Delta}_{\textbf{OF}})$ is always nonnegative in this case. The non-existence of any negative value of $(\bar{\Delta}_{\textbf{OR}}-\bar{\Delta}_{\textbf{OF}})$ implies that forgetting measurement outcomes always gives better precision than remembering whatever the choice of the values of parameters $\alpha$ and $\beta$ are. Hence, when the parameter to be estimated is in the coefficient of the traceless matrix of two-outcome qubit POVM elements, remembering measurement outcomes is not good for precision at all. (b) The right plot shows $(\bar{\Delta}_{\textbf{OR}}-\bar{\Delta}_{\textbf{OF}})$ along the vertical axis, during estimation of $\alpha$, as the ratio $\beta/\alpha$ is changed along the horizontal axis, for various values of the parameter $\alpha$. There are many values of $\alpha$ and $\beta$ such that $(\bar{\Delta}_{\textbf{OR}}-\bar{\Delta}_{\textbf{OF}})$ have negative values, implying that keeping encoding measurement outcomes can provide better precision. This is altered if the value of $\alpha$ is lowered. Therefore, remembering encoding measurement outcomes can provide better precision in such scenarios only if the parameter of interest is in the coefficient of the nonzero trace part of POVM elements and it is a necessary criterion, not a sufficient one. Both axes, as well as \(\alpha\), of both panels, represent dimensionless quantities.
  }
    \label{fig:1} 
\end{figure*}

\section{
Quantum metrology of a  parameter encoded by a measurement}
\label{single parameter}
In this section, we restrict ourselves to quantum estimation of a parameter encoded by any arbitrary two-outcome qubit measurement using qubit probes. 
We address the question:
\emph{does recording (remembering) the encoding measurement outcomes enhance the precision of estimating the measurement parameter in a large variety of such scenarios compared to ignoring them?}
Hence, we ask whether the estimation errors satisfy, $\Delta_{\textbf{OR}} \leq \Delta_{\textbf{OF}} ,$
i.e., is the OR error universally smaller than or equal to the OF error?

The most general two-outcome qubit POVM, $\{E_i\}$ with $i \in \{1,2\}$, can be written as 
\begin{align}
\label{eqn-measurement}
    E_1&=\alpha \mathbb{I}_2+\beta \sigma_{\hat{n}}, \quad
    E_2=(1-\alpha)\mathbb{I}_2-\beta \sigma_{\hat{n}},
\end{align}
where $\alpha,\beta$ are real constants with $\beta \leq \alpha$ and $\hat{n}$ is a unit vector defining the measurement direction through $\sigma_{\hat{n}} \coloneqq \hat{n} \cdot \vec{\sigma}$. We refer to $\alpha$ and $\beta$ as the ``measurement coefficients".
Any one of the variables $\alpha$, or $\beta$, or $\hat{n}$ can serve as the parameter of interest, which we denote by $x_1$. 

For single measurement–parameter estimation of the POVM in Eq.~\eqref{eqn-measurement}, three cases are possible: (i) estimation of either $\alpha$, $\beta$, or a parameter on which both $\alpha$ and $\beta$ depend, (ii) estimation of measurement direction $\hat{n}$, and
(iii) estimation of a parameter that jointly governs the measurement coefficient and/or direction, i.e., a parameter on which $\alpha$ and $\hat{n}$ depend, or $\beta$ and $\hat{n}$ depend, or all three $\alpha$, $\beta$, and $\hat{n}$ depend.
In this measurement-parameter estimation task, we compare the OR error with the OF error and arrive at the following statement:

\begin{theorem}
\label{theorem:0}
    When the eigenstates of the encoded qubit state do not depend on the parameter to be estimated, then keeping encoding measurement records can yield better precision in estimation of the measurement parameter only if the coefficient of the nonzero-trace operator of a two-outcome qubit POVM elements carries the parameter dependence. 
\end{theorem}

\begin{proof}
We consider the case where the eigenstates of the encoded qubit state do not depend on the parameter to be estimated. We prove the theorem in two steps: first, we show that when the encoded parameter is the coefficient of the traceless operator of the POVM elements, keeping records of the measurement outcomes never improves the precision compared to ignoring the outcomes; second, we demonstrate by an explicit example that considering the coefficient of the nonzero-trace part of the POVM elements as the parameter of interest can allow remembering the encoding measurement outcomes to yield better precision than forgetting the outcomes.

Without any loss of generality, we set $\hat{n} \coloneqq (0,0,1)$. 
Any general qubit input state can be written as
\begin{align}
\label{eq:general_input_state}
    \rho_{\text{in}}=\frac{1}{2}\left(\mathbb{I}_2+r\sigma_{\hat{r}}\right), 
\end{align}
where $\sigma_{\hat{r}} \coloneqq \hat{r} \cdot \vec{\sigma}$ and $\vec{r} = r \hat{r}$ denotes the Bloch vector with norm $r \in [0,1]$ and unit vector $\hat{r} \coloneqq (\sin \phi_1 \cos \phi_2, \sin \phi_1 \sin \phi_2, \cos \phi_1)$ with $\phi_1 \in [0, \pi)$ and $\phi_2 \in [0, 2\pi)$.
\\

\textbf{Step I:} We estimate the coefficient of the traceless operator of POVM
elements in Eq.~\eqref{eqn-measurement}, $\beta$. For convenience, we set $\beta=x_1$.

Let us start with the scenario where the encoding measurement outcomes are recorded. 
If the input probe is pure, i.e., $r=1$, the OR error is given by
\begin{align*}
    \Delta_{\textbf{OR}}=&\frac{\sqrt{\alpha^2-x_1^2}(\alpha+x_1\cos\phi_1)^2}{\alpha\sin\phi_1}\notag \\&+\frac{\sqrt{(1-\alpha)^2-x_1^2}(1-\alpha-x_1\cos\phi_1)^2}{(1-\alpha)\sin\phi_1}.
\end{align*}

On the other hand, when the encoding measurement outcomes are forgotten, the encoded state is given by
\begin{align}
\label{eq:forgot_measurement_encoding}
    &\rho_F=\frac{1}{2}[(1+r\cos(\phi_1))\ketbra{0}{0}+(1-r\cos(\phi_1))\ketbra{1}{1} \nonumber \\
&+r\sin(\phi_1)h(x_1,\beta)\left(e^{i\phi_2}\ketbra{0}{1}+e^{-i\phi_2}\ketbra{1}{0}\right)],
\end{align}
where $\{\ket{0},\ket{1}\}$ denotes the standard computational basis and $h(\alpha,x_1) \coloneqq \left(\sqrt{\alpha^2-x_1^2}+\sqrt{(1-\alpha)^2-x_1^2}\right).$
The QFI, corresponding to the encoded state $\rho_{F}$ in this $\beta$ estimation task, is
\begin{equation*}
    \left(\frac{h_2r\sin\phi_1}{2}\right)^2\left(1+\frac{3+hr^2}{1-hr^2}\sin^2\phi_1\right),
\end{equation*}
where $h_2\coloneqq \partial h(\alpha,x_1)/\partial x_1$.
\\
Note that QFI is maximum over all input states when $r=1$ and $\sin\phi_1=1$, and therefore the maximum QFI becomes $h_2^2/(1-h^2)$.
The minimum OF error is obtained when the QFI is maximum, i.e., when $r=1$ and $\sin\phi_1=1$.
Consequently, the corresponding minimum OF error over all input states is
\begin{align}
\label{eqn-outcome-forgotten-specific}
    \bar{\Delta}_{\textbf{OF}}=\frac{\sqrt{\alpha^2-x_1^2}\sqrt{(1-\alpha)^2-x_1^2}\sqrt{1-h(\alpha,x_1)^2}}{x_1 h(\alpha,x_1)}.
\end{align}
If we minimize $\Delta_{\textbf{OR}}$ over all pure input states, it can be easily seen that $\Delta_{\textbf{OF}}\leq\Delta_{\textbf{OR}}$ always occurs, which implies that forgetting encoding measurement outcomes always yields better precision compared to recorded ones in this $\beta$ estimation task.
Furthermore, even if we minimize $\Delta_{\textbf{OR}}$ in this $\beta$ estimation task over all probes, we have obtained the same conclusion (which is checked numerically).
In Fig.~\ref{fig:1}(a), we plot $\bar{\Delta}_{\textbf{OF}} - \bar{\Delta}_{\textbf{OR}}$ corresponding to this $\beta$ estimation task for several values of $\alpha$ as a function of $\beta/\alpha$, and see  that the inequality $\Delta_{\textbf{OF}}\leq\Delta_{\textbf{OR}}$ is always satisfied.

Using the property of QFI given in Eq,~\eqref{property:QFI-1}, we conclude that the relation  $\bar{\Delta}_{\textbf{OF}}\leq\bar{\Delta}_{\textbf{OR}}$ continue to hold, even when $\beta$ is any function of $x_1$, i.e., $\beta=\beta(x_1)$, rather than $\beta= x_1$ and the behavior persists even after optimizing over all probes.
\\

\textbf{Step II:} 
We will provide here an example of input state by which we show that when the coefficient of the nonzero-trace part of the POVM elements acts as the parameter of interest, remembering the encoding measurement outcomes can yield better precision than forgetting one. So, we estimate the parameter $\alpha$ in Eq.~\eqref{eqn-measurement} and for convenience, we set $\alpha=x_1$.

In this $\alpha$ estimation task, when the encoding measurement outcomes are ignored, the QFI corresponding to the encoded state $\rho_F$ is calculated to be 
$$\left(\frac{h_1r\sin\phi_1}{2}\right)^2\left(1+\frac{3+hr^2}{1-hr^2}\sin^2\phi_1\right),$$ 
which is independent of $\phi_2$. 
Here $h(x_1,\beta) \coloneqq \left(\sqrt{x_1^2-\beta^2}+\sqrt{(1-x_1)^2-\beta^2}\right)$, and $h_1 \coloneqq \partial h(x_1,\beta)/\partial x_1$ 
This QFI is maximum over all parameters of input states, $r \in [0,1]$ and $\phi_1 \in [0, \pi/2]$, at $r=1$ and $\phi_1=\pi$. 
Therefore, when one estimates $\alpha$ and forgets the encoding measurement outcomes, the maximum QFI over all probes is $h_1^2/(1-h^2)$.
The corresponding minimum $\Delta_{\textbf{OF}}$ is 
\begin{equation}
\label{eq:beta_estimation_qfi_forgotten_measurement}
   \bar{\Delta}_{\textbf{OF}} =  \frac{\sqrt{1-\left(\sqrt{x_1^2-\beta^2}+\sqrt{(1-x_1)^2-\beta^2}\right)^2}}{\left[\frac{x_1}{\sqrt{x_1^2-\beta^2}}-\frac{1-x_1}{\sqrt{(1-x_1)^2-\beta^2}}\right]}.
\end{equation}

In the remembering encoding measurement outcomes scenario, let us choose the input state as $\ket{+} \coloneqq \frac{1}{\sqrt{2}}(\ket{0} +\ket{1})$. The OR error for the input state $\ket{+}$ is given by
\begin{equation}
\label{eq:avgerr_remember_outcome}
 \Delta_{\textbf{OR}} =   \frac{x_1^2\sqrt{x_1^2-\beta^2}+(1-x_1)^2\sqrt{(1-x_1)^2-\beta^2}}{\beta}.
\end{equation}
Note that the OR error in Eq.~\eqref{eq:avgerr_remember_outcome} is less than the $\bar{\Delta}_{\textbf{OF}}$ in Eq.~\eqref{eq:beta_estimation_qfi_forgotten_measurement} when $x_1=0.4$ and $\beta=0.1$.
Since $\bar{\Delta}_{\textbf{OF}}$ in Eq.~\eqref{eq:beta_estimation_qfi_forgotten_measurement} is the OF error obtained by minimizing over all input states, the OR error corresponding to $\ket{+}$ state will always be less than the $\bar{\Delta}_{\textbf{OF}}$ in Eq.~\eqref{eq:beta_estimation_qfi_forgotten_measurement} when $x_1=0.4$ and $\beta=0.1$.
Therefore, at $x_1=0.4$ and $\beta=0.1$, if we optimize the OR error over all input states, the OR error will show either better and better or same precision compared to the optimum OF error in this $\alpha$ estimation task. Thus, when the coefficient of the nonzero-trace part of the POVM elements acts as the parameter of interest, remembering the encoding measurement outcomes can yield better precision than forgetting one. The property of QFI mentioned in Eq.~\eqref{property:QFI-1} implies that even if $\alpha$ is a function of $x_1$, the relation  $\Delta_{\textbf{OR}} \leq \bar{\Delta}_{\textbf{OF}}$ still holds when $x_1=0.4$ and $\beta=0.1$ and the input state is $\ket{+}$ while calculating the OR error.
However, in general, it may not be true for all $\alpha$ and $\beta$ in this $\alpha$ estimation task. 
See Fig.~\ref{fig:1}(b), we plot $\bar{\Delta}_{\textbf{OF}} - \bar{\Delta}_{\textbf{OR}}$ for several values of $\alpha$ as a function of $\beta/\alpha$, and notice that there exist a few $\alpha$ for which $\bar{\Delta}_{\textbf{OR}} \leq \bar{\Delta}_{\textbf{OF}}$ is not true as well, which implies that our theorem is a necessary criterion, not a sufficient one.

Thus, the proof is completed. 
\end{proof}

Furthermore, we examine whether keeping the encoding measurement records can provide better precision than ignoring the outcomes in the case where the coefficient of the zero trace operator of the POVM elements in Eq.~\eqref{eqn-measurement} carries the parameter of interest, while the eigenstates of the encoded qubit depend on the parameter to be estimated, unlike the statement in theorem~\ref{theorem:0}. To investigate this, we should consider scenarios in which $\beta$ and $\theta$ are arbitrary functions of $x_1$, where $x_1$ denotes the parameter of interest. We perform numerical calculations of $\bar{\Delta}_{\mathbf{OF}}$ and $\bar{\Delta}_{\mathbf{OR}}$,  taking $\theta=x_1$ and $\beta=\alpha\cos(x_1)$ for different values of $\alpha$ that are independent of $x_1$. We plot $\bar{\Delta}_{\textbf{OF}}-\bar{\Delta}_{\textbf{OR}}$ as a function of $\theta$ in Fig.~\ref{fig:alpha_est} and we find that
$\bar{\Delta}_{\textbf{OF}} \leq \bar{\Delta}_{\textbf{OR}}$
in all cases. This implies that Theorem~\ref{theorem:0} also holds here: keeping the records of the encoding measurement outcomes never improves precision when the coefficient of zero trace operator of POVM elements carries dependence of the parameter of interest, even though the eigenstates of the encoded qubit state depend on the parameter to be estimated.

At this point, it is interesting to remember a fact about a form of entanglement called bound entanglement~\cite{bound-entanglement-prl-1998}. It has been argued~\cite{shor_smolin_2001PRL} that mixing two quantum states, each with zero distillable entanglement but nonzero entanglement and having negative partial transpose, can produce a state with nonzero distillable entanglement, under a certain assumption. This effect is analogous to our result that forgetting the outcome of an encoding measurement can be useful in metrology.

\begin{figure}
\includegraphics[scale=0.38]{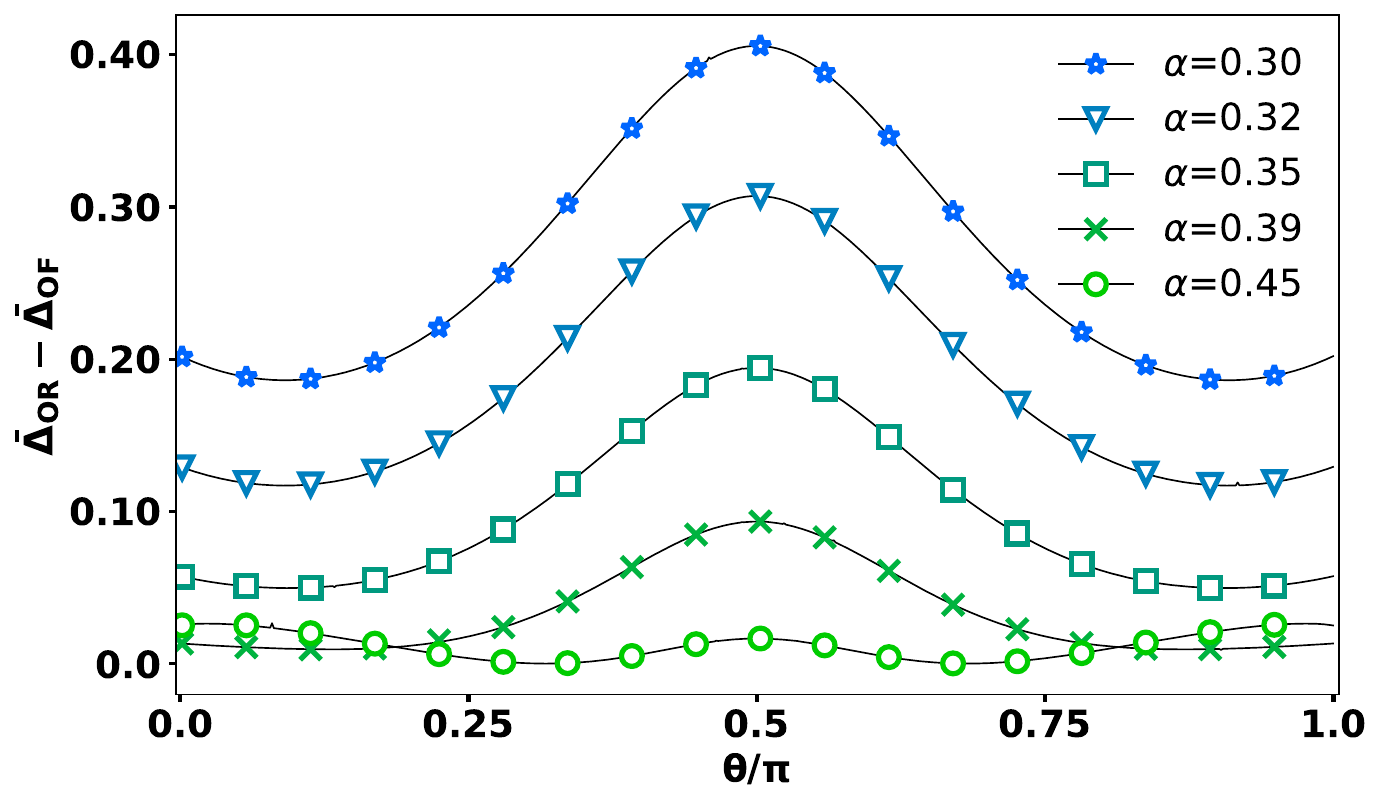}
    \caption{The quantity $(\bar{\Delta}_{\textbf{
    OR}}-\bar{\Delta}_{\textbf{
    OF}})$ is plotted along the vertical axis in case of estimating $\theta$, which is varied along the horizontal axis in the unit of $\pi$. Various values of $\alpha$ are considered with $\beta=\alpha\cos\theta$ is taken in this plot. 
    Note that $(\bar{\Delta}_{\textbf{
    OR}}-\bar{\Delta}_{\textbf{
    OF}})$ never has negative values when eigenstates of the encoded state depend on the parameter to be estimated and hence, it supports the generalization of our Theorem.~\ref{theorem:0}. Both axes, as well as \(\alpha\), of this plot, represent dimensionless quantities.}
\label{fig:alpha_est}
\end{figure}

\vspace{2mm}

\section{Criterion for the invertibility of the achievable quantum Fisher information matrix}
\label{multiple-parameters-general}
The multiparameter QCRB in any quantum multiparameter estimation task is useful when both the achievability of the QCRB and the nonsingularity of the QFIM hold. In this direction, Ref.~\cite{ghosh2025efficient} provides a necessary and sufficient criterion for the QFIM to be noninvertible in the estimation of any two parameters encoded by an arbitrary quantum process using qubit probes. However, since the multiparameter QCRB is not always attainable even when the QFIM is invertible, it is important to find the criterion that ensures both QFIM invertibility and multiparameter QCRB achievability.
Keeping this in mind, we identify below the criteria for achieving QCRB and non-singularity of QFIM associated with achievable QCRB, for estimation task of any multiparameter encoding involving qubit probe.

\begin{theorem}
\label{corollary:achivability}
For any multiparameter estimation task encoded by an arbitrary quantum process using qubit probes, the QCRB corresponding to estimation of parameters in the set $\bm{x}\coloneqq\{x_i\}$ is achievable if and only if $\bra{\mathbf{e}_0}\frac{\partial\rho(\bm{x})}{\partial x_i}\ket{\mathbf{e}_1}$ is a real multiple of $\bra{\mathbf{e}_0}\frac{\partial\rho(\bm{x})}{\partial x_j}\ket{\mathbf{e}_1}$ for all $x_i$ and $x_j$, 
where $\rho(\bm{x})$ denotes the encoded state with the eigenbasis $\{\ket{\mathbf{e}_0}, \ket{\mathbf{e}_1}\}$.
\end{theorem}

\begin{proof}
 
Any general qubit encoded state, $\rho(\bm{x})$, can be written in the spectral decomposition form as $\rho(\bm{x})=\sum_{\mu = 0}^{1} \lambda_u\ketbra{e_\mu}{e_\mu}$, where $\{\ket{e_\mu}\}$ denotes the eigenbasis corresponding to the eigenvalue $\{\lambda_\mu\}$.
It was known that multiparameter
QCRB is achievable if and only if the expectation values of the
commutators of the SLDs corresponding to all pairs of parameters of interest vanish with respect to the encoded state, i.e., the Uhlmann matrix elements vanish~\cite{PhysRevA.89.023845,PhysRevA.94.052108}. 
The Uhlmann matrix elements for the encoded state $\rho(\bm{x})$ can be written as
\begin{widetext}
\begin{align}
    [\mathcal{D}]_{ij}=4\sum_{\substack{\mu, \nu \\ \lambda_{\mu}+\lambda_{\nu}\neq0}}\lambda_{\mu}\frac{\bra{e_\mu}\frac{\partial\rho(\bm{x})}{\partial x_i}\ket{e_\nu}\bra{e_\nu}\frac{\partial\rho(\bm{x})}{\partial x_j}\ket{e_\mu}-\bra{e_\mu}\frac{\partial\rho(\bm{x})}{\partial x_j}\ket{e_\nu}\bra{e_\nu}\frac{\partial\rho(\bm{x})}{\partial x_i}\ket{e_\mu}}{(\lambda_{\mu}+\lambda_{\nu})^2},
\end{align}
\end{widetext}
with $i,j \in \{1,2\}$.
After simplification, $[\mathcal{D}]_{ij}$ for the encoded state $\rho(\bm{x})$, looks like the following:
\begin{align*}
    &[\mathcal{D}]_{ij}\\&=2(\lambda_0-\lambda_1)\text{Im}\left(\bra{e_0}\frac{\partial\rho(\bm{x})}{\partial x_i}\ket{e_1}\bra{e_1}\frac{\partial\rho(\bm{x})}{\partial x_j}\ket{e_0}\right),
\end{align*}
since $\lambda_0 + \lambda_1 = 1$ and when $\mu=\nu$, terms in the matrix elements of Uhlmann matrix get zero.

Note that $\lambda_1=\lambda_0$ is a trivial solution of $[\mathcal{D}]_{ij} = 0$ for all $i,j$, because the encoded state is a maximally mixed state when $\lambda_1=\lambda_0$. The corresponding Uhlmann matrix elements are by definition zero and the achievability of QCRB is automatically satisfied. Excluding this trivial case, the achievability of QCRB is satisfied only when 
\begin{equation*}
\text{Im}\left(\bra{e_0}\frac{\partial\rho(\bm{x})}{\partial x_i}\ket{e_1}\bra{e_1}\frac{\partial\rho(\bm{x})}{\partial x_j}\ket{e_0}\right)=0.
\end{equation*}
It implies that the two-parameter QCRB is achievable for an arbitrary quantum encoding process using qubit probes if and only if 
\begin{equation}
\label{eq: qubit achievability condition}
    \bra{e_0}\frac{\partial\rho(\bm{x})}{\partial x_i}\ket{e_1}=c_{ij} \bra{e_0}\frac{\partial\rho(\bm{x})}{\partial x_j}\ket{e_1}.
\end{equation}
where $c_{ij}$ denotes a real number. 
It completes the proof.
\end{proof}

In the following theorem, considering two-parameter estimation using qubit probes, we state necessary and sufficient condition regarding the invertibility of QFIM corresponding to achievable QCRB.

\begin{theorem}
\label{theorem1}
For any estimation task of two parameters $\{x_1,x_2\}$ encoded by an arbitrary quantum process using qubit probes, the quantum Fisher information matrix associated with an achievable quantum Cram\'er–Rao bound is invertible if and only if $\bra{e_0}\frac{\partial\rho(x_1,x_2)}{\partial x_1}\ket{e_1}$ is a real constant multiple of $\bra{e_0}\frac{\partial\rho(x_1,x_2)}{\partial x_2}\ket{e_1}$ and $\frac{\partial \rho(x_1,x_2)}{\partial x_1}$ does not commute with $\frac{\partial \rho(x_1,x_2)}{\partial x_2}$, where $\rho(x_1, x_2)$ denotes the encoded state with the eigenstates $\{\ket{e_0},\ket{e_1}\}$.
\end{theorem}

\begin{proof}
The proof of saturability is already given in the proof of Theorem~\ref{theorem0}. Now, we are ready to find the criterion on singularity of QFIM corresponding to achievable QCRB in such scenarios.
Since $\lambda_\mu = \bra{e_\mu}\rho(x_1, x_2)\ket{e_\mu}$, $\sum_\mu \lambda_\mu = 1$, and the derivative of the density matrices are traceless, we have
\begin{equation*}
    \bra{e_1}\frac{\partial\rho(x_1, x_2)}{\partial x_i}\ket{e_1}=-\bra{e_0}\frac{\partial\rho(x_1, x_2)}{\partial x_i}\ket{e_0}; \quad i \in \{1,2\}.
\end{equation*}
The achievability condition of QCRB in Eq.~\eqref{eq: qubit achievability condition} implies that
\begin{equation}\label{eq:qubit_two_parameter_achievability}
   \bra{e_0}\frac{\partial\rho(x_1, x_2)}{\partial x_2}\ket{e_1}=c_{21} s_1,
\end{equation}
where $s_i=\bra{e_0}\frac{\partial\rho}{\partial x_i}\ket{e_1}$.
Therefore, in the eigenbasis of the encoded state,  SLD operators for different parameters of interest can be written as
\begin{align*}
L_{x_1}=\begin{pmatrix}
\frac{u_1}{\lambda_1} & s_{1} \\
s_{1}^* & \frac{-u_1}{\lambda_2}
\end{pmatrix}, \quad L_{x_2}=\begin{pmatrix}
\frac{u_2}{\lambda_1} & c_{21} s_{1} \\
c_{21} s_{1}^* & \frac{-u_2}{\lambda_2}
\end{pmatrix}.
\end{align*}
Here $u_i=\bra{e_0}\frac{\partial\rho}{\partial x_i}\ket{e_0}$.
If one of the eigenvalues of the encoded state vanishes, the corresponding matrix elements of $L_{x_i}$ are set to zero.
Substituting the expressions of SLD operators in QFIM in Eq.~\eqref{eq:QFIM with SLD}, we find that $\det\left(\mathbb{F}_Q\right)=0 \Rightarrow u_2=c_{21}u_1$. Combining this condition with the achievability condition of QCRB in Eq.~\eqref{eq:qubit_two_parameter_achievability} gives us
\begin{equation}\label{eq:derivatives_multiple_of_each_other}
    \frac{\partial \rho(x_1, x_2)}{\partial x_2}=c_{21} \frac{\partial \rho(x_1, x_2)}{\partial x_1}.
\end{equation}
Together with the fact that the derivative of density matrices are traceless, Eq.~\eqref{eq:derivatives_multiple_of_each_other} tells that $\frac{\partial \rho(x_1, x_2)}{\partial x_1}$ commutes with $\frac{\partial \rho(x_1, x_2)}{\partial x_2}$.
It completes the proof.
\end{proof}
In the next theorem, we extend the two-parameter scenrio to multiple parameter case.

\begin{theorem}
\label{theorem0}
For any estimation task involving multiple parameters encoded by an arbitrary quantum process using qubit probes, the quantum Fisher information matrix associated with an achievable quantum Cram\'er–Rao bound is not invertible if number of parameters to be estimated is greater than two.
\end{theorem}

\begin{proof}
As the simplest situation of Eq.~\eqref{eq: qubit achievability condition} when the encoded state is a real Hermitian matrix, the achievability condition of QCRB in Eq.~\eqref{eq: qubit achievability condition} is automatically satisfied.

The QFIM, incorporating this saturability condition, is easy to calculate. The elements of QFIM have the expression
\begin{equation*}
    [\mathbb{F}_Q]_{ij}=s_{i}s_j^{*}+\frac{u_{i}}{\sqrt{\lambda_1\lambda_2}}\frac{u_j^{*}}{\sqrt{\lambda_1\lambda_2}}.
\end{equation*}
The key point to note is, this matrix is sum of two rank one matrix (outer product of vectors). So $\mathbb{F}_Q$ can have maximum rank $2$, whereas the dimension of $\mathbb{F}_Q$ is $N\times N$, with $N$ being the number of parameters. This fact implies $\det[\mathbb{F}_Q]=0$ for $N>2.$ Which means $\mathbb{F}_Q$ is singular if number of parameter to estimate, exceeds $2.$

Thus, the proof is completed.
\end{proof}

Note that Theorem~\ref{theorem0} tells when multiparameter QCRB is not a proper bound to use. 

\section{Estimation of multiple parameters encoded by a measurement}
\label{multiple-parameters-measurement}
In many quantum multiparameter estimation scenarios, the QCRB ceases to be a valid metrological bound for two reasons discussed in Sec.~\ref{pre}, e.g.,
the simultaneous estimation of any two unitarily encoded parameters using qubit probes, where the QCRB achievability constraint always forces the QFIM singular~\cite{Paris2024}.
Motivated by this we examine whether similar limitations arise in simultaneous multiple measurement–encoded parameter estimation. Here, we focus on simultaneous estimation of multiple parameters for the most general two-outcome qubit POVM using qubit probes.
Throughout these analyses, we set $\hat{n} \coloneqq (\sin\theta, 0, \cos\theta)$ with $\theta \in [0, 2\pi)$ for the sake of simplicity, even though a general unit vector $\hat{n}$ is parameterized by two angular parameters. 
For the general two-outcome POVM in Eq.~\eqref{eqn-measurement}, we consider four multiple measurement-encoded parameter estimation tasks using qubit probes: simultaneous estimation of (i) $\alpha$ and $\beta$, (ii) $\alpha$ and $\theta$, (iii) $\beta$ and $\theta$, and (iv) all three parameters $\alpha$, $\beta$, and $\theta$. 
For each scenario, we analyze two types of estimation errors: the OR error and the OF error.

Let us suppose the parameters are encoded by the measurement $\{E_i\}$ (in Eq.~\eqref{eqn-measurement}) into the general qubit input state (in Eq.~\eqref{eq:general_input_state}). If the encoding measurement outcomes are recorded, then the normalized encoded states are $\rho_{f_1}$ and $\rho_{f_2}$ respectively, when $E_1$ and $E_2$ get clicked, obtained with 
probabilities $Q_1 \coloneqq \alpha+r\beta(\hat{n}\cdot\hat{r})$, and $Q_2 \coloneqq 1-\alpha-r\beta(\hat{n}\cdot\hat{r})$,
where
\begin{align*}
    &\rho_{f_1}=\frac{1}{2}\left(\mathbb{I}_2+\frac{rg_1}{Q_1}\sigma_{\hat{r}}+\frac{\beta+r(\alpha-g_1)(\hat{n}\cdot\hat{r})}{Q_1}\sigma_{\hat{n}}\right),\\
    &\rho_{f_2}=\frac{1}{2}\left(\mathbb{I}_2+\frac{rg_2}{Q_2}\sigma_{\hat{r}}-\frac{\beta-r(1-\alpha-g_2)(\hat{n}\cdot\hat{r})}{Q_2}\sigma_{\hat{n}}\right).
\end{align*}
Here $g_1 \coloneqq \sqrt{\alpha^2-\beta^2}$, $g_2 \coloneqq \sqrt{(1-\alpha)^2-\beta^2}$. On the other hand, when encoding measurement outcomes are ignored, the encoded state is given by
\begin{align}\label{eq:forgotten_state}
    \rho_F=\frac{1}{2}\big[\mathbb{I}_2+r(\hat{n}\cdot \hat{r})\sigma_{\hat{n}}+rh(\alpha,\beta)\left[\sigma_{\hat{r}}-(\hat{n}\cdot \hat{r})\sigma_{\hat{n}}\right]\big]
\end{align}
with $h(\alpha,\beta) \coloneqq \left(\sqrt{\alpha^2-\beta^2}+\sqrt{(1-\alpha)^2-\beta^2}\right).$
\\
\\
\textbf{Case I (Estimation of $\alpha$ and $\beta$):} 
We study the simultaneous estimation of the two measurement coefficients, $\alpha$ and $\beta$, of the general two-outcome qubit POVM given in Eq.~\eqref{eqn-measurement}.

\begin{theorem}
    For the simultaneous estimation of two measurement coefficients of a general qubit POVM, the quantum Fisher information matrix associated with any achievable quantum Cram\'er–Rao bound is always singular, irrespective of whether the encoding measurement outcomes are recorded or ignored.
\end{theorem}

\begin{proof}
For convenience, we set $\alpha = x_1 $ and $\beta = x_2$.

To analyze the case where measurement outcomes are recorded, we set $\hat{n}=(0,0,1)$ without loss of generality.
Let the Bloch vector of the input state be $\vec{r}=(r_x,r_y,r_z)$ with $r_x,r_y,r_z\in[0,1]$ and $r_x^2+r_y^2+r_z^2\leq1$.
The encoded state corresponding to an $E_1$ click, $\rho_{f_1}$, can then be simplified in terms of $(r_x,r_y,r_z)$ as 
\begin{align*}
    &\rho_{f_1}=\frac{1}{2}\left(
    \mathbb{I}_2+\frac{x_2+r_z x_1}{x_1+x_2 r_z}\sigma_{\hat{z}}\right)
    +\frac{1}{2}\frac{\sqrt{x_1^2-x_2^2}}{x_1+x_2 r_z}\left(r_x\sigma_{\hat{x}}+r_y\sigma_{\hat{y}}\right).
\end{align*}
And the corresponding partial derivative of the encoded state $\rho_{f_1}$ with respect to $x_1$ and $x_2$ are given by
\begin{align*}
    &\frac{\partial \rho_{f_1}}{\partial x_1}=\frac{x_2}{(x_1+x_2 r_z)^2}A,\quad \frac{\partial \rho_{f_1}}{\partial x_2}=-\frac{x_1}{(x_1+x_2 r_z)^2}A;\notag \\
    &A\coloneqq \frac{x_2+x_1 r_z}{\sqrt{x_1^2-x_2^2}}(r_x\sigma_{\hat{x}}+r_y\sigma_{\hat{y}})-(1-r_z^2)\sigma_{\hat{z}}.
\end{align*}
Note that $\frac{\partial \rho_{f_1}}{\partial x_1}$ commutes with $\frac{\partial \rho_{f_1}}{\partial x_2}$ whatever the parameters of interest $x_1, x_2$ and the form of input state are. 
The same result holds as well when $E_2$ clicks.
Therefore, according to Theorem.~\ref{theorem1}, the QFIM corresponding to the achievable QCRB must be singular in the simultaneous estimation of $\alpha$ and $\beta$ when the encoding measurement outcomes are recorded.

Also, it is very easy to see that when the encoding measurement outcomes are forgotten, the derivatives of the encoded state in Eq.~\eqref{eq:forgotten_state}, with respect to encoded parameters, will have the form 
\begin{align*}
    \frac{\partial\rho_F}{\partial x_i}=\frac{rh_{x_i}}{2}\mathbf{L}, 
\end{align*}
where $h_{x_i} \coloneqq (\partial h(x_1,x_2)/\partial x_i)$, and $\mathbf{L} \coloneqq \left[\sigma_{\hat{r}}-(\hat{n}\cdot \hat{r})\sigma_{\hat{n}}\right]$.
Therefore, like the encoding measurement outcome remembering scenario, $[\frac{\partial\rho_F}{\partial x_1}, \frac{\partial\rho_F}{\partial x_2}]  = 0 $ holds automatically as well whatever the parameters of interest $x_1, x_2$ and the form of input state are. 
Hence, the QFIM corresponding to the achievable QCRB is singular in this case, according to Theorem.~\ref{theorem1}.

Thus, it completes the proof.
\end{proof}

From now on, we study the simultaneous estimation of multiple measurement encoding parameters, which includes the parameter  $\theta$ related to measurement direction, such that with $\theta\in[0,2\pi)$ the measurement direction can be written as $\hat{n}\coloneqq (\cos\theta,0,\sin\theta)$ without loss of generality.
\\
\\
\textbf{Case II (Estimation of $\alpha$ and $\theta$):} 
Here, we study the simultaneous estimation of two parameters, $\alpha$ and $\theta$, of the measurement given in Eq.~\eqref{eqn-measurement}. Our goal is to examine whether the QFIM associated with the achievable QCRB becomes singular or not in both the encoding measurement outcome-recorded and outcome-forgotten scenarios, similar to the case of simultaneous estimation of the encoding measurement coefficients $\alpha$ and $\beta$. For convenience, we set $\alpha = x_1 $ and $\theta = x_2$.

Let us now start with the encoding measurement outcome–recorded scenario. We choose the Bloch vector of the input state as $\vec{r}= r\,\hat{r}$ and define $\hat{r}\coloneqq -\partial\hat{n}/\partial\theta$, where $r\in[0,1]$ denotes the norm of the Bloch vector of the input state and $\hat{n}\coloneqq (\cos\theta,0,\sin\theta)$ denotes the unit vector characterizing the measurement direction with $\theta\in[0,2\pi)$ and we can choose such $\hat{n}$ without loss of generality. With this choice, $\hat{n}\cdot\hat{r}=0$.

When the measurement, $\{E_i\}$, is performed on the input state and outcomes $E_1$ and $E_2$ click, the corresponding encoded states are
\begin{align*}
    \rho_{f_1}&=\frac{1}{2}\left[\mathbb{I}_2+r\sqrt{1-\frac{\beta^2}{x_1^2}}\,\sigma_{\hat{r}}+\frac{\beta}{\alpha}\sigma_{\hat{n}}\right],\\
    \rho_{f_2}&=\frac{1}{2}\left[\mathbb{I}_2+r\sqrt{1-\frac{\beta^2}{(1-x_1)^2}}\,\sigma_{\hat{r}}-\frac{\beta}{1-x_1}\sigma_{\hat{n}}\right],
\end{align*}
respectively.
It is straightforward to see that $\partial\rho_{f_1}/\partial x_1$ and $\partial\rho_{f_1}/\partial x_2$ are hermitian matrices with real elements, due to the chosen orientations of the measurement direction and the input-state Bloch vector. Therefore, by Corollary~\ref{corollary:achivability}, the QCRB corresponding to $\rho_{f_1}$ is achievable for arbitrary values of the Bloch vector norm and parameters $x_1$, $\beta$, and $x_2$, provided the vectors are chosen in this manner. By the same reasoning, the QCRB corresponding to $\rho_{f_2}$ is also achievable.
Furthermore, since $\hat{n}\cdot\hat{r}=0$, we have
$\left[\frac{\partial \rho_{f_1}}{\partial x_1},\frac{\partial \rho_{f_1}}{\partial x_2}\right]=0$
if and only if $r^2=1,$
and similarly,
$\left[\frac{\partial \rho_{f_2}}{\partial x_1},\frac{\partial \rho_{f_2}}{\partial x_2}\right]=0$ if and only if $r^2=1$.
Hence, whenever the input state is not pure, the QFIM associated with the achievable QCRB is invertible (nonsingular), independent of the Bloch-vector norm and the parameter values, as long as the Bloch vector and measurement direction satisfy the chosen orthogonality condition. Therefore, in the simultaneous estimation of $\alpha$ and $\theta$, there exist many choices of input states and values of $\alpha, \beta$, and $\theta$ for which the QCRB is achievable and valid when the encoding measurement outcome is recorded.

Let us now look at the simultaneous estimation case of $\alpha$ and $\theta$, when measurement outcomes are forgotten. 
We choose the Bloch vector of the input state as $\vec{r}\coloneqq (r_x,0,r_z)$ with $r_x\geq0$, $r_z\geq0$ and $r_x^2+r_y^2\leq1$. The encoded state with these choices, in the encoding measurement outcome forgotten scenarios, becomes
\begin{align*}
    \rho_{F}=\frac{1}{2}\left[\mathbb{I}_2+\tilde{r}_x(x_1,x_2)\sigma_{x}+\tilde{r}_z(x_1,x_2)\sigma_z\right].
\end{align*}
where
\begin{align*}
    &\tilde{r}_x(x_1,x_2) \coloneqq r_xh+(1-h)(r_z\cos x_2+r_x\sin x_2)\sin x_2,\\
    &\tilde{r}_z(x_1,x_2) \coloneqq r_zh+(1-h)(r_z\cos x_2+r_x\sin x_2)\cos x_2,
\end{align*}
and $h(x_1,\beta) \coloneqq \left(\sqrt{x_1^2-\beta^2}+\sqrt{(1-x_1)^2-\beta^2}\right)$.
Since $\partial\rho_{F}/\partial x_1$ and $\partial\rho_{F}/\partial x_2$ are Hermitian matrices with real elements, due to the chosen orientations of the measurement direction and the input-state Bloch vector, by Corollary~\ref{corollary:achivability}, the QCRB corresponding to $\rho_{F}$ is automatically achievable for arbitrary values of the Bloch vector norm and parameters $x_1$, $\beta$, and $x_2$, provided they are chosen in this manner.
The commutation condition
$\left[\frac{\partial \rho_{F}}{\partial x_1},\frac{\partial \rho_{F}}{\partial x_2}\right]=0$
holds if and only if
\begin{align*}
    \frac{\partial \tilde{r}_x/\partial x_1}{\partial \tilde{r}_z/\partial x_1}
    =
    \frac{\partial \tilde{r}_x/\partial x_2}{\partial \tilde{r}_z/\partial x_2},
\end{align*}
which simplifies to
\begin{equation}
\label{naam_na_jana_pakhi}
    \tan x_2=\tan(\phi_1-2x_2),
\end{equation}
where $r_x/r_z\coloneqq \tan\phi_1$.
Note that there exist many choices of $r_x$, $r_z$, and $x_2$ for which Eq.~\eqref{naam_na_jana_pakhi} is not satisfied, and therefore it follows from Theorem.~\ref{theorem1} that the QFIM corresponding to the achievable QCRB is nonsingular for a broad class of input states and parameters in the simultaneous estimation of $\alpha$ and $\theta$ when the encoding measurement outcomes are forgotten.

Furthermore, note that when the encoding measurement $\{E_i\}$ is of the form
\begin{align*}
    E_1 &= \alpha \mathbb{I}_2+\beta\sin\theta\,\sigma_x+\beta\cos\theta\,\sigma_z, \\
    E_2 &= (1-\alpha)\mathbb{I}_2-\beta\sin\theta\,\sigma_x-\beta\cos\theta\,\sigma_z,
\end{align*}
with $\theta\in[0,2\pi)\setminus\{\pi/4,3\pi/4\}$, and the input state chosen as
\[
\rho_{\text{in}}=\frac{1}{2}\Big[\mathbb{I}_2+r\cos\theta\,\sigma_x-r\sin\theta\,\sigma_z\Big],
\]
with $r<1$, the QFIM is invertible whenever the QCRB is achievable in the simultaneous estimation of $\alpha$ and $\theta$, in both the outcome-remembered and outcome-forgotten scenarios, irrespective of the values of $\alpha$ and $\beta \leq \alpha$.
\\
\\
\textbf{Case III (Estimation of $\beta$ and $\theta$):} 
In the simultaneous estimation of $\beta$ and $\theta$, the achievability of QCRB is readily satisfied  when the same parameter regime of input probe and the set of parameters of interest are chosen as in Case II.

Let us define the SLD operators corresponding to $\beta$ and $\theta$ by $L_\alpha$ and $L_\beta$ respectively. From the result obtained in Case II,
the SLD operators corresponding to $\alpha$ and $\beta$ can be written as 
\begin{align}\label{eq: alpha and beta interdependent}
    L_{\alpha}=c_{\alpha}R, \quad L_{\beta}=c_{\beta}R,
\end{align}
where $c_{\alpha}$ and $c_{\beta}$ are two real numbers and $R$ is a matrix. 
Assuming $L_{\theta}$ to be the SLD corresponding to parameter $\theta$, the QFIM corresponding to the simultaneous estimation of $\alpha$ and $\theta$ then can be written as
\begin{align*}
    \begin{pmatrix}
c_{\alpha}^2 \Tr\left[\rho R^2\right] & c_{\alpha}\Tr\left[\rho \frac{RL_{\theta}+L_{\theta}R}{2}\right] \\
c_{\alpha}\Tr\left[\rho \frac{RL_{\theta}+L_{\theta}R}{2}\right] & \Tr\left[\rho L_{\theta}^2\right]
\end{pmatrix}
\end{align*} and the QFIM corresponding to the simultaneous estimation of $\beta$ and $\theta$ then can be written as
\begin{align*}
    \begin{pmatrix}
c_{\beta}^2 \Tr\left[\rho R^2\right] & c_{\beta}\Tr\left[\rho \frac{RL_{\theta}+L_{\theta}R}{2}\right] \\
c_{\beta}\Tr\left[\rho \frac{RL_{\theta}+L_{\theta}R}{2}\right] & \Tr\left[\rho L_{\theta}^2\right]
\end{pmatrix}.
\end{align*}
Note that the invertibility of the QFIM corresponding to the simultaneous estimation of $\alpha$ and $\theta$ automatically guarantees the invertibility of the QFIM corresponding to the simultaneous estimation of $\beta$ and $\theta$ if the trivial case $c_{\beta}=0$ is not considered, which can be avoided
by some choices of $\alpha$ and $\beta$.
We have seen in Case II that, for some parameters of interest values and input states, the QFIM corresponding to simultaneous estimation of $\alpha$ and $\theta$ is invertible irrespective of remembering or forgetting measurement outcomes. These arguments, combined with the achievability of QCRB in simultaneous estimation $\beta$ and $\theta$, we conclude that there are at least some choices of parameters and input states for which QCRB corresponding to $\beta$ and $\theta$ is valid and achievable and it holds true independent of remembering or forgetting encoding measurement outcomes.
\\
\\
\textbf{Case IV (Estimation of $\alpha$, $\beta$ and $\theta$):} In the estimation of  $\alpha$, $\beta$, and $\theta$ simultaneously, 
the QFIM is always singular 
due to the intrinsic interdependence between $\alpha$ and $\beta$ (discussed in Case II), in either
the outcome-remembered or outcome-forgotten scenarios. This follows directly from Eq.~\eqref{eq: alpha and beta interdependent}, which holds for any choice of input state and a set of parameter values.

\section{Conclusion}
\label{conclusion}
In this work, we studied 
the quantum metrology of parameters encoded by a quantum measurement. We considered two scenarios in the estimation of a parameter or a set of parameters characterizing a general two-outcome qubit measurement: one is when the measurement outcomes are known and another one is when the outcomes are ignored.

We found that forgetting the outcomes provides better precision in a large variety of estimation scenarios when a parameter is encoded by a two-outcome qubit measurement. We also derived a necessary criterion on achieving precision advantage in the outcome-recorded scenarios compared to the forgotten ones in such estimation scenarios.
Moreover, we established a necessary and sufficient criterion for the simultaneous estimation of multiple parameters, encoded by an arbitrary quantum process using qubit probes, 
such that the achievable QCRB is a valid bound. Furthermore, for simultaneous estimation of two 
parameters encoded via two-outcome qubit measurement, not related to measurement direction, the QFIM associated with the achievable QCRB is always singular, irrespective of whether the measurement outcomes are recorded or ignored, mirroring the behavior found in estimation of two unitarily-encoded parameters with qubit probes. 
In contrast, when the two simultaneously estimated parameters contain one parameter related to measurement direction, the QFIM need not be singular in either the outcome-remembered or outcome-forgotten scenarios.

\vspace{0.1cm}
\section*{Acknowledgments}
We acknowledge the use of Armadillo. 

\bibliography{sharpness_est}

@article{Paris2024,
  title = {Dimension matters: precision and incompatibility in multi-parameter quantum estimation models},
  volume = {9},
  ISSN = {2058-9565},
  url = {http://dx.doi.org/10.1088/2058-9565/ad7498},
  DOI = {10.1088/2058-9565/ad7498},
  number = {4},
  journal = {QST},
  publisher = {IOP Publishing},
  author = {Candeloro,  Alessandro and Pazhotan,  Zahra and Paris,  Matteo G A},
  year = {2024},
  month = sep,
  pages = {045045}
}

@article{PhysRevLett.72.3439,
  title = {Statistical distance and the geometry of quantum states},
  author = {Braunstein, S. L. and Caves, C. M.},
  journal = {Phys. Rev. Lett.},
  volume = {72},
  issue = {22},
  pages = {3439--3443},
  numpages = {0},
  year = {1994},
  month = {May},
  publisher = {American Physical Society},
  doi = {10.1103/PhysRevLett.72.3439},
  url = {https://link.aps.org/doi/10.1103/PhysRevLett.72.3439}
}

@article{PhysRevLett.96.010401,
  title = {Quantum Metrology},
  author = {Giovannetti, Vittorio and Lloyd, Seth and Maccone, Lorenzo},
  journal = {Phys. Rev. Lett.},
  volume = {96},
  issue = {1},
  pages = {010401},
  numpages = {4},
  year = {2006},
  month = {Jan},
  publisher = {American Physical Society},
  doi = {10.1103/PhysRevLett.96.010401},
  url = {https://link.aps.org/doi/10.1103/PhysRevLett.96.010401}
}

@Article{Giovannetti2011,
author={Giovannetti, Vittorio
and Lloyd, Seth
and Maccone, Lorenzo},
title={Advances in quantum metrology},
journal={Nat. Photon.},
year={2011},
month={Apr},
day={01},
volume={5},
number={4},
pages={222},
issn={1749-4893},
doi={10.1038/nphoton.2011.35},
url={https://doi.org/10.1038/nphoton.2011.35}
}

@article{PhysRevA.110.012620,
  title = {Enhancing precision of atomic clocks by tuning disorder in accessories},
  author = {Bhattacharyya, Aparajita and Ghoshal, Ahana and Sen, Ujjwal},
  journal = {Phys. Rev. A},
  volume = {110},
  issue = {1},
  pages = {012620},
  numpages = {16},
  year = {2024},
  month = {Jul},
  publisher = {American Physical Society},
  doi = {10.1103/PhysRevA.110.012620},
  url = {https://link.aps.org/doi/10.1103/PhysRevA.110.012620}
}

@article{bhattacharyya2025quantumsensingevenversus,
  title = {Quantum sensing of even- versus odd-body interactions},
  author = {Bhattacharyya, Aparajita and Saha, Debarupa and Sen, Ujjwal},
  journal = {Phys. Rev. A},
  volume = {112},
  issue = {3},
  pages = {L030603},
  numpages = {8},
  year = {2025},
  month = {Sep},
  publisher = {American Physical Society},
  url = {https://link.aps.org/doi/10.1103/3f8n-b1kp}
}

@book{holevo,
    author = {A S Holevo},
    title = {Probabilistic and statistical aspects of quantum theory, Vol. 1},
    publisher = {Springer Science \& Business Media},
    year = {2011}
}

@article{agarwal2025quantumsensingultracoldsimulators,
      title={Quantum sensing with ultracold simulators in lattice and ensemble systems: a review}, 
      author={Keshav Das Agarwal and Sayan Mondal and Ayan Sahoo and Debraj Rakshit and Sen (De), Aditi and Ujjwal Sen},
      year={2025},
      journal={arXiv:2507.06348},
      url={https://arxiv.org/abs/2507.06348}, 
}

@article{saha2025entanglementconstrainedquantummetrologyrapid,
      title={Entanglement-Constrained Quantum Metrology: Rapid Low-Entanglement Gains, Tapered High-Level Growth}, 
      author={Debarupa Saha and Ujjwal Sen},
      year={2025},
      journal={arXiv:2507.03512},
      url={https://arxiv.org/abs/2507.03512}, 
}

@article{bhattacharyya2025precisionestimatingindependentlocal,
      title={Precision in estimating independent local fields: attainable bound and indispensability of genuine multiparty entanglement}, 
      author={Aparajita Bhattacharyya and Ujjwal Sen},
      year={2024},
      journal={arXiv:2407.20142},
      url={https://arxiv.org/abs/2407.20142}, 
}

@article{mondal2025optimalquantumprecisionnoise,
      title={Optimal quantum precision in noise estimation: Is entanglement necessary?}, 
      author={Shuva Mondal and Priya Ghosh and Ujjwal Sen},
      year={2025},
     journal={arXiv:2507.22413},
      url={https://arxiv.org/abs/2507.22413}, 
}

@article{pal2025rolephaseoptimalprobe,
      title={Role of phase of optimal probe in noncommutativity vs coherence in quantum multiparameter estimation}, 
      author={Ritopriyo Pal and Priya Ghosh and Ahana Ghoshal and Ujjwal Sen},
      year={2025},
      journal={arXiv:2507.04824},
      url={https://arxiv.org/abs/2507.04824}, 
}

@article{mondal2024multicriticalquantumsensorsdriven,
      title={Multicritical quantum sensors driven by symmetry-breaking}, 
      author={Sayan Mondal and Ayan Sahoo and Ujjwal Sen and Debraj Rakshit},
      year={2024},
      journal={arXiv:2407.14428},
      url={https://arxiv.org/abs/2407.14428}, 
}

@article{agarwal2025criticalquantummetrologyusing,
      title={Critical quantum metrology using non-Hermitian spin model with RT-symmetry}, 
      author={Keshav Das Agarwal and Tanoy Kanti Konar and Leela Ganesh Chandra Lakkaraju and Sen (De), Aditi},
      year={2025},
      journal={arXiv:2503.24331},
      url={https://arxiv.org/abs/2503.24331}, 
}

@article{Agrawal2025indefinitetime,
  doi = {10.22331/q-2025-07-03-1785},
  url = {https://doi.org/10.22331/q-2025-07-03-1785},
  title = {Indefinite {T}ime {D}irected {Q}uantum {M}etrology},
  author = {Agrawal, Gaurang and Halder, Pritam and Sen (De), Aditi},
  journal = {{Quantum}},
  issn = {2521-327X},
  publisher = {{Verein zur F{\"{o}}rderung des Open Access Publizierens in den Quantenwissenschaften}},
  volume = {9},
  pages = {1785},
  month = jul,
  year = {2025}
}

@article{RevModPhys.90.035005,
  title = {Quantum metrology with nonclassical states of atomic ensembles},
  author = {Pezz\`e, Luca and Smerzi, Augusto and Oberthaler, Markus K. and Schmied, Roman and Treutlein, Philipp},
  journal = {Rev. Mod. Phys.},
  volume = {90},
  issue = {3},
  pages = {035005},
  numpages = {70},
  year = {2018},
  month = {Sep},
  publisher = {American Physical Society},
  doi = {10.1103/RevModPhys.90.035005},
  url = {https://link.aps.org/doi/10.1103/RevModPhys.90.035005}
}

@article{
doi:10.1126/science.1104149,
author = {Vittorio Giovannetti  and Seth Lloyd  and Lorenzo Maccone },
title = {Quantum-Enhanced Measurements: Beating the Standard Quantum Limit},
journal = {Science},
volume = {306},
number = {5700},
pages = {1330-1336},
year = {2004},
doi = {10.1126/science.1104149},
URL = {https://www.science.org/doi/abs/10.1126/science.1104149}}

@article{RevModPhys.89.035002,
  title = {Quantum sensing},
  author = {Degen, C. L. and Reinhard, F. and Cappellaro, P.},
  journal = {Rev. Mod. Phys.},
  volume = {89},
  issue = {3},
  pages = {035002},
  numpages = {39},
  year = {2017},
  month = {Jul},
  publisher = {American Physical Society},
  doi = {10.1103/RevModPhys.89.035002},
  url = {https://link.aps.org/doi/10.1103/RevModPhys.89.035002}
}

@article{RevModPhys.90.035006,
  title = {Quantum-enhanced measurements without entanglement},
  author = {Braun, Daniel and Adesso, Gerardo and Benatti, Fabio and Floreanini, Roberto and Marzolino, Ugo and Mitchell, Morgan W. and Pirandola, Stefano},
  journal = {Rev. Mod. Phys.},
  volume = {90},
  issue = {3},
  pages = {035006},
  numpages = {47},
  year = {2018},
  month = {Sep},
  publisher = {American Physical Society},
  doi = {10.1103/RevModPhys.90.035006},
  url = {https://link.aps.org/doi/10.1103/RevModPhys.90.035006}
}

@Article{Helstrom1969,
author={Helstrom, Carl W.},
title={Quantum detection and estimation theory},
journal={J. Stat. Phys.},
year={1969},
month={Jun},
day={01},
volume={1},
number={2},
pages={231-252},
issn={1572-9613},
doi={10.1007/BF01007479},
url={https://doi.org/10.1007/BF01007479}
}

@article{HOLEVO1973337,
title = {Statistical decision theory for quantum systems},
journal = {JMVA},
volume = {3},
number = {4},
pages = {337-394},
year = {1973},
issn = {0047-259X},
doi = {https://doi.org/10.1016/0047-259X(73)90028-6},
url = {https://www.sciencedirect.com/science/article/pii/0047259X73900286},
author = {A S Holevo}
}

@Article{Pirandola2018,
author={Pirandola, S.
and Bardhan, B. R.
and Gehring, T.
and Weedbrook, C.
and Lloyd, S.},
title={Advances in photonic quantum sensing},
journal={Nat. Photon.},
year={2018},
month={Dec},
day={01},
volume={12},
number={12},
pages={724-733},
issn={1749-4893},
doi={10.1038/s41566-018-0301-6},
url={https://doi.org/10.1038/s41566-018-0301-6}
}

@article{PhysRevLett.94.020502,
  title = {Entanglement Assisted Metrology},
  author = {Cappellaro, P. and Emerson, J. and Boulant, N. and Ramanathan, C. and Lloyd, S. and Cory, D. G.},
  journal = {Phys. Rev. Lett.},
  volume = {94},
  issue = {2},
  pages = {020502},
  numpages = {4},
  year = {2005},
  month = {Jan},
  publisher = {American Physical Society},
  doi = {10.1103/PhysRevLett.94.020502},
  url = {https://link.aps.org/doi/10.1103/PhysRevLett.94.020502}
}

@article{chaki2025nonpositivemeasurementsarentbeneficial,
      title={Non-positive measurements aren't beneficial in quantum metrology for unitary encoding, but can be for open schemes}, 
author={Paranjoy Chaki and Debarupa Saha and Kornikar Sen and Ujjwal Sen},
      year={2025},   
    journal={arXiv:2509.24585},
      url={https://arxiv.org/abs/2509.24585}, 
}

@article{Mehboudi2019-thermometry,
doi = {10.1088/1751-8121/ab2828},
url = {https://doi.org/10.1088/1751-8121/ab2828},
year = {2019},
month = {jul},
publisher = {IOP Publishing},
volume = {52},
number = {30},
pages = {303001},
author = {Mehboudi, Mohammad and Sanpera, Anna and Correa, Luis A},
title = {Thermometry in the quantum regime: recent theoretical progress},
journal = {J. Phys. A: Math. Theor.}
}

@article{sarkar2025,
  title={Comparing physical quantities with finite-precision: beyond standard metrology and an illustration for cooling in quantum processes},
  author={Sarkar, Anindita and Chaki, Paranjoy and Ghosh, Priya and Sen, Ujjwal},
  journal={arXiv:2510.24484},
  year={2025},
  url = {
https://doi.org/10.48550/arXiv.2510.24484
}  
}

@article{PhysRevA.63.042304,
  title = {Quantum channel identification problem},
  author = {Fujiwara, Akio},
  journal = {Phys. Rev. A},
  volume = {63},
  issue = {4},
  pages = {042304},
  numpages = {4},
  year = {2001},
  month = {Mar},
  publisher = {American Physical Society},
  doi = {10.1103/PhysRevA.63.042304},
  url = {https://link.aps.org/doi/10.1103/PhysRevA.63.042304}
}

@article{PRXQuantum.2.010343,
  title = {Asymptotic Theory of Quantum Channel Estimation},
  author = {Zhou, Sisi and Jiang, Liang},
  journal = {PRX Quantum},
  volume = {2},
  issue = {1},
  pages = {010343},
  numpages = {25},
  year = {2021},
  month = {Mar},
  publisher = {American Physical Society},
  doi = {10.1103/PRXQuantum.2.010343},
  url = {https://link.aps.org/doi/10.1103/PRXQuantum.2.010343}
}

@article{MasahitoHayashi_2010,
doi = {10.1088/1742-6596/233/1/012016},
url = {https://dx.doi.org/10.1088/1742-6596/233/1/012016},
year = {2010},
month = {jun},
publisher = {},
volume = {233},
number = {1},
pages = {012016},
author = {Hayashi, Masahito},
title = {Quantum channel estimation and asymptotic bound},
journal = {J. Phys. Conf. Ser.}
}

@article{Akio_Fujiwara_2003,
doi = {10.1088/0305-4470/36/29/314},
url = {https://dx.doi.org/10.1088/0305-4470/36/29/314},
year = {2003},
month = {jul},
publisher = {},
volume = {36},
number = {29},
pages = {8093},
author = {Fujiwara, Akio and Imai, Hiroshi},
title = {Quantum parameter estimation of a generalized Pauli channel},
journal = {J. Phys. A Math. Gen.}
}

@article{rafal_2023,
  title = {Using Adaptiveness and Causal Superpositions Against Noise in Quantum Metrology},
  author = {Kurdzia\l{}ek, Stanis\l{}aw and G\'orecki, Wojciech and Albarelli, Francesco and Demkowicz-Dobrza\ifmmode \acute{n}\else \'{n}\fi{}ski, Rafa\l{}},
  journal = {Phys. Rev. Lett.},
  volume = {131},
  issue = {9},
  pages = {090801},
  numpages = {7},
  year = {2023},
  month = {Aug},
  publisher = {American Physical Society},
  doi = {10.1103/PhysRevLett.131.090801},
  url = {https://link.aps.org/doi/10.1103/PhysRevLett.131.090801}
}

@article{rafal_2014,
  title = {Using Entanglement Against Noise in Quantum Metrology},
  author = {Demkowicz-Dobrza\ifmmode \acute{n}\else \'{n}\fi{}ski, Rafal and Maccone, Lorenzo},
  journal = {Phys. Rev. Lett.},
  volume = {113},
  issue = {25},
  pages = {250801},
  numpages = {5},
  year = {2014},
  month = {Dec},
  publisher = {American Physical Society},
  doi = {10.1103/PhysRevLett.113.250801},
  url = {https://link.aps.org/doi/10.1103/PhysRevLett.113.250801}
}

@article{demkowicz_2012,
  title={The elusive Heisenberg limit in quantum-enhanced metrology},
  author={Demkowicz-Dobrza{\'n}ski, Rafa{\l} and Ko{\l}ody{\'n}ski, Jan and Gu{\c{t}}{\u{a}}, M{\u{a}}d{\u{a}}lin},
  journal={Nat. comm.},
  volume={3},
  number={1},
  pages={1063},
  year={2012},
  publisher={Nature Publishing Group UK London},
  url={http://dx.doi.org/10.1038/ncomms2067}
}

@article{Pirandola17,
  title = {Ultimate Precision of Adaptive Noise Estimation},
  author = {Pirandola, Stefano and Lupo, Cosmo},
  journal = {Phys. Rev. Lett.},
  volume = {118},
  issue = {10},
  pages = {100502},
  numpages = {6},
  year = {2017},
  month = {Mar},
  publisher = {American Physical Society},
  doi = {10.1103/PhysRevLett.118.100502},
  url = {https://link.aps.org/doi/10.1103/PhysRevLett.118.100502}
}

@article{RAZAVIAN2019825,
title = {Quantum metrology out of equilibrium},
journal = {Physica A: Statistical Mechanics and its Applications},
volume = {525},
pages = {825-833},
year = {2019},
issn = {0378-4371},
doi = {https://doi.org/10.1016/j.physa.2019.03.125},
url = {https://www.sciencedirect.com/science/article/pii/S0378437119303607},
author = {Sholeh Razavian and Matteo G.A. Paris}
}

@article{Huelga97,
  title = {Improvement of Frequency Standards with Quantum Entanglement},
  author = {Huelga, S. F. and Macchiavello, C. and Pellizzari, T. and Ekert, A. K. and Plenio, M. B. and Cirac, J. I.},
  journal = {Phys. Rev. Lett.},
  volume = {79},
  issue = {20},
  pages = {3865--3868},
  numpages = {0},
  year = {1997},
  month = {Nov},
  publisher = {American Physical Society},
  doi = {10.1103/PhysRevLett.79.3865},
  url = {https://link.aps.org/doi/10.1103/PhysRevLett.79.3865}
}

@article{Alipur14,
  title = {Quantum Metrology in Open Systems: Dissipative Cram\'er-Rao Bound},
  author = {Alipour, S. and Mehboudi, M. and Rezakhani, A. T.},
  journal = {Phys. Rev. Lett.},
  volume = {112},
  issue = {12},
  pages = {120405},
  numpages = {6},
  year = {2014},
  month = {Mar},
  publisher = {American Physical Society},
  doi = {10.1103/PhysRevLett.112.120405},
  url = {https://link.aps.org/doi/10.1103/PhysRevLett.112.120405}
}

@Article{Yousefjani17,
author={Yousefjani, R.
and Salimi, S.
and Khorashad, A. S.},
title={Noisy metrology: a saturable lower bound on quantum Fisher information},
journal={QIP},
year={2017},
month={Apr},
day={17},
volume={16},
number={6},
pages={144},
url={https://doi.org/10.1007/s11128-017-1596-9}
}

@article{Peng23,
      title={Dissipative quantum Fisher information for a general Liouvillian parameterized process}, 
      author={Jia-Xin Peng and Baiqiang Zhu and Weiping Zhang and Keye Zhang},
      year={2023},
      journal={arXiv: 2308.10183},
      url={https://arxiv.org/abs/2308.10183}, 
}

@article{Bhattacharyya24a,
  title = {Restoring metrological quantum advantage of measurement precision in a noisy scenario},
  author = {Bhattacharyya, Aparajita and Ghoshal, Ahana and Sen, Ujjwal},
  journal = {Phys. Rev. A},
  volume = {109},
  issue = {5},
  pages = {052626},
  numpages = {17},
  year = {2024},
  month = {May},
  publisher = {American Physical Society},
  doi = {10.1103/PhysRevA.109.052626},
  url = {https://link.aps.org/doi/10.1103/PhysRevA.109.052626}
}

@article{ghosh2025efficient,
  title={Efficient Estimation of Multiple Temperatures via a Collisional Model},
  author={Ghosh, Srijon and Chakraborty, Sagnik and Franco, Rosario Lo},
  journal={arXiv:2511.20448},
  year={2025},
url = {
https://doi.org/10.48550/arXiv.2511.20448}
}

@article{shor_smolin_2001PRL,
  title = {Nonadditivity of Bipartite Distillable Entanglement Follows from a Conjecture on Bound Entangled Werner States},
  author = {Shor, Peter W. and Smolin, John A. and Terhal, Barbara M.},
  journal = {Phys. Rev. Lett.},
  volume = {86},
  issue = {12},
  pages = {2681},
  numpages = {0},
  year = {2001},
  month = {Mar},
  publisher = {American Physical Society},
  doi = {10.1103/PhysRevLett.86.2681},
  url = {https://link.aps.org/doi/10.1103/PhysRevLett.86.2681}
}

@article{Carollo2018,
  title = {Uhlmann curvature in dissipative phase transitions},
  volume = {8},
  ISSN = {2045-2322},
  url = {http://dx.doi.org/10.1038/s41598-018-27362-9},
  DOI = {10.1038/s41598-018-27362-9},
  number = {1},
  journal = {Sci. Rep.},
  publisher = {Springer Science and Business Media LLC},
  author = {Carollo,  Angelo and Spagnolo,  Bernardo and Valenti,  Davide},
  year = {2018},
  month = jun 
}

@book{bell_aspect_2004, 
place={Cambridge}, 
edition={2}, 
title={Speakable and Unspeakable in Quantum Mechanics: Collected Papers on Quantum Philosophy}, 
publisher={Cambridge University Press}, 
author={Bell, J. S.}, 
year={2004}
}

@article{bell-review,
  title = {{B}ell nonlocality},
  author = {Brunner, N. and Cavalcanti, D. and Pironio, S. and Scarani, V. and Wehner, S.},
  journal = {Rev. Mod. Phys.},
  volume = {86},
  issue = {2},
  pages = {419},
  numpages = {60},
  year = {2014},
  month = {Apr},
  publisher = {American Physical Society},
  url = {https://link.aps.org/doi/10.1103/RevModPhys.86.419}
}

@article{Bennett1993,
  title = {Teleporting an unknown quantum state via dual classical and Einstein-Podolsky-Rosen channels},
  author = {Bennett, Charles H. and Brassard, Gilles and Cr\'epeau, Claude and Jozsa, Richard and Peres, Asher and Wootters, William K.},
  journal = {Phys. Rev. Lett.},
  volume = {70},
  issue = {13},
  pages = {1895--1899},
  numpages = {0},
  year = {1993},
  month = {Mar},
  publisher = {American Physical Society},
  doi = {10.1103/PhysRevLett.70.1895},
  url = {https://link.aps.org/doi/10.1103/PhysRevLett.70.1895}
}

@Article{Pirandola2015,
author={Pirandola, S.
and Eisert, J.
and Weedbrook, C.
and Furusawa, A.
and Braunstein, S. L.},
title={Advances in quantum teleportation},
journal={Nat. Photon.},
year={2015},
month={Oct},
day={01},
volume={9},
number={10},
pages={641-652},
issn={1749-4893},
doi={10.1038/nphoton.2015.154},
url={https://doi.org/10.1038/nphoton.2015.154}
}

@article{bound-entanglement-prl-1998,
  title = {Mixed-State Entanglement and Distillation: Is there a ``Bound'' Entanglement in Nature?},
  author = {Horodecki, Micha\l{} and Horodecki, Pawe\l{} and Horodecki, Ryszard},
  journal = {Phys. Rev. Lett.},
  volume = {80},
  issue = {24},
  pages = {5239--5242},
  numpages = {0},
  year = {1998},
  month = {Jun},
  publisher = {American Physical Society},
  doi = {10.1103/PhysRevLett.80.5239},
  url = {https://link.aps.org/doi/10.1103/PhysRevLett.80.5239}
}

@article{PhysRevA.89.023845,
  title = {Tradeoff in simultaneous quantum-limited phase and loss estimation in interferometry},
  author = {Crowley, Philip J. D. and Datta, Animesh and Barbieri, Marco and Walmsley, I. A.},
  journal = {Phys. Rev. A},
  volume = {89},
  issue = {2},
  pages = {023845},
  numpages = {9},
  year = {2014},
  month = {Feb},
  publisher = {American Physical Society},
  doi = {10.1103/PhysRevA.89.023845},
  url = {https://link.aps.org/doi/10.1103/PhysRevA.89.023845}
}

@article{PhysRevA.94.052108,
  title = {Compatibility in multiparameter quantum metrology},
  author = {Ragy, Sammy and Jarzyna, Marcin and Demkowicz-Dobrza\ifmmode \acute{n}\else \'{n}\fi{}ski, Rafa\l{}},
  journal = {Phys. Rev. A},
  volume = {94},
  issue = {5},
  pages = {052108},
  numpages = {11},
  year = {2016},
  month = {Nov},
  publisher = {American Physical Society},
  doi = {10.1103/PhysRevA.94.052108},
  url = {https://link.aps.org/doi/10.1103/PhysRevA.94.052108}
}

@article{Liu_2020,
url = {https://dx.doi.org/10.1088/1751-8121/ab5d4d},
year = {2019},
month = {dec},
publisher = {IOP Publishing},
volume = {53},
number = {2},
pages = {023001},
author = {Liu, J. and Yuan, H. and Lu, X. and Wang, X.},
title = {Quantum Fisher information matrix and multiparameter estimation},
journal = {J. Phys. A: Math. Theor.}
}

\section*{appendix}
\label{section:quantum_metrology}
This section offers a detailed discussion of quantum metrology.

Quantum metrology~\cite{Liu_2020} typically involves three stages. In the first stage, a suitable probe state is prepared. In the second stage, the probe undergoes an evolution that encodes the parameters of interest. In the final stage, measurements are performed on the encoded state, and the unknown parameters are estimated from the resulting data.

In more detail, let us start with the single quantum parameter estimation theory. In the first step, an input state $\rho_{\text{in}}\in \mathcal{H}_d$ is prepared, a state of our choice. In the second step, let us consider a scenario where a parameter $x_1 \in \Xi$, to be estimated, is encoded into the probe $\rho_{\text{in}}$ by a quantum process, $\Lambda_{x_1}(\bullet)$.
After the encoding, we have the encoded state 
\begin{equation*}
    \rho(x_1) \coloneqq \Lambda_{x_1}(\rho_{\text{in}}).
\end{equation*}
In the last step, let us suppose that a quantum measurement $\{\Pi_y\}$ is performed on the encoded state where $y$ denotes the measurement outcome. 
In order to infer about the parameter of interest, a function $\hat{\xi}:\mathbb{R}\to \Xi$ called an estimator is introduced, which is a map from the set of real numbers to the parameter space. 
An estimator is said to be locally unbiased if it satisfies the following conditions:
\begin{align*}
    & \int dy \Tr{(\Pi_y \rho(x_1))} \hat{\xi}(y)=\tilde{x}_1,\\
    &\int dy\, \hat{\xi}(y) \frac{d \Tr{(\Pi_y \rho_{x_1})}}{dy}\bigg|_{x_1=\tilde{x}_1}=1.
\end{align*}
where $\tilde{x}_1$ denotes the actual value of the parameter of interest.

Here, we quantify the estimation error for the parameter of interest by the mean square error of the unbiased estimator,
    $V(x_1) \coloneqq \int dy \left(\hat{\xi}(y)-\tilde{x}_1 \right)^2 \Tr{\left(\Pi_y \rho(x_1)\right)}$.
When the measurement is fixed, the mean squared error (MSE) of any unbiased estimator is lower bounded by a quantity called the Cram\'er-Rao bound, i.e., 
\begin{equation*}
    V(x_1) \geq \frac{1}{\mathbb{F}_C(\Pi_y,\rho(x_1))}.
\end{equation*}
which is obtained by optimizing over all unbiased estimators for a fixed measurement. $\mathbb{F}_C(\Pi_y,\rho(x_1))$ denotes the Classical Fisher Information (CFI) corresponding to the measurement $\{\Pi_y\}$ and it is defined as
\begin{equation}
    \mathbb{F}_C(\Pi_y,\rho(x_1)) \coloneqq \int dy \frac{\Tr{\left[\Pi_y \dot{\rho}({x_1})\right]^2}}{\Tr{\left[\Pi_y \rho({x_1})\right]}},
\end{equation}
where $\dot{\rho}(x_1)$ denotes the derivative of the encoded state with respect to the parameter of interest. 

Now optimizing the CFI over all quantum measurements, we obtain the QCRB, i.e.,
\begin{equation}
    \Delta^2\hat{\xi}\geq \frac{1}{\mathbb{F}_C(\Pi_y,\rho(x_1))} \geq \frac{1}{\mathbb{F}_Q(\rho(x_1))} 
\end{equation}
The QFI is defined in terms of SLD operator and the encoded state as follows:
\begin{equation}
    \mathbb{F}_Q(\rho(x_1)) \coloneqq \Tr{(\rho({x_1}) L^2_{x_1})}.
\end{equation}
where $L_{x_1}$ denotes the SLD operator
corresponding to the parameter of interest. It is a hermitian operator whose matrix element in the eigenbasis of the encoded state is given by
\begin{equation*}
    [L_{x_1}]_{ij}=2\frac{\langle e_{\mu}|\frac{\partial \rho({x_1})}{\partial \theta}|e_{\nu}\rangle}{\lambda_{\mu}+\lambda_{\nu}},
\end{equation*}
where $|e_{\mu}\rangle$ is the eigenvector of $\rho({x_1})$ corresponding to the eigenvalue $\lambda_{\mu}$. In the above representation, at least one of the indices must be restricted to the support of the encoded state $\rho({x_1})$ or else the denominator goes to $0$. When both indices are out of the support, the element can be considered to be $0$ by definition.  The QCRB is saturated, i.e, $\mathbb{F}_C(\Pi_y,\rho(x_1))=\mathbb{F}_Q(\rho(x_1))$ when the projective measurement in the eigenbasis of the SLD operator is performed.

Before going to the discussion of multiparameter estimation, we want to mention a useful property of QFI. If the QFI corresponding to parameter $x_1$ is known to be $\mathbb{F}_Q(\rho(x_1))$ then the QFI corresponding to another parameter $z_1$, which is a function of $x_1$ known as $z_1=f(x_1)$, then the QFI corresponding to the new parameter of interest $z_1$ is 
\begin{equation*}
\label{property:QFI-1} \mathbb{F}_Q(\rho(z_1))=\frac{\mathbb{F}_Q(\rho(x_1))}{(\partial z_1/\partial x_1)^2}. 
\end{equation*}

The multiparameter estimation theory is the generalization of the single-parameter estimation theory. Let us suppose that the parameters of interest are $\bm{x} \coloneqq \{x_1,x_2, \ldots, x_m\}$. These parameters are encoded into a quantum state $\rho_{\text{in}}$ through a quantum process $\Lambda_{\bm{x}}(\bullet)$, yielding the encoded state $\rho({\bm{x}})\coloneqq \Lambda_{\bm{x}}\left[\rho_{\text{in}}\right]$.
The values of all encoded parameters are inferred from the statistics of the measurement outcomes of measurement $\{\Pi_y\}$. To perform this inference, $m$ functions, ${\bm{\hat{\xi}}(y)}\coloneqq\left(\hat{\xi}_1,\hat{\xi}_2, \ldots, \hat{\xi}_m\right)$, —called estimators—are introduced, where $\hat{\xi}_i$ denotes the locally unbiased estimator that corresponds to $i$-th parameter of interest. 
In this simultaneous estimation of $m$ parameters scenario, the estimation error is quantified by mean square error matrix (MSEM) of the $m$ estimators as follows:
\begin{align*}   
V (\bm{x}) \coloneqq\sum_{y} \Tr{[\Pi_y \rho(\bm{x})]} \left(\bm{\hat{\xi}} (y)-\bm{\tilde{x}}\right)\left(\hat{\bm{\xi}}(y)-\bm{\tilde{x}}\right)^T,
\end{align*}
where $\bm{\tilde{x}} \coloneqq \{\tilde{x}_1, \tilde{x}_2, \ldots, \tilde{x}_m\}$ denotes the vector of the actual values of parameters of interest.

The MESM of $m$ locally unbiased estimators corresponding to the $m$ parameters satisfies the following lower bound:
\begin{equation*}
   V (\bm{x})   \geq \mathbb{\mathbf{F}}^{-1}_\mathbf{C}(\Pi_y, \rho(\bm{x}))  \geq \mathbb{\mathbf{F}}^{-1}_\mathbf{Q}(\rho(\bm{x})).
\end{equation*}
The classical Fisher information matrix, $\mathbb{\mathbf{F}}_\mathbf{C}(\Pi_y, \rho(\bm{x}))$, is defined as
\begin{align*}
    &\big[\mathbb{\mathbf{F}}_\mathbf{C} (\Pi_y, \rho(\bm{x}))\big]_{ij} \\&\coloneqq \int dy \frac{\Tr{\left[\frac{1}{2}\Pi_y \{L_{x_i},\rho({\bm{x}})\}\right]}\Tr\left[{\frac{1}{2}\Pi_y\{L_{x_j},\rho({\bm{x}})\}}\right]}{\Tr{\left[\Pi_y \rho({\bm{x}})\right]}}.
\end{align*} 
The QFIM
is the matrix analog of the QFI for a single parameter. The elements of the QFIM are defined as
\begin{equation*}  \left[\mathbb{\mathbf{F}}_\mathbf{Q}(\rho(\bm{x}))\right]_{ij}  \coloneqq\,\frac{1}{2}\Tr{\left[\rho({\bm{x}})\{L_{x_i}L_{x_j}+L_{x_j} L_{x_i}\}\right]}, 
\end{equation*}
where $i,j\in \{1,2,\ldots ,m\}$ and $L_{x_i}$ denotes the SLD operator corresponding to the $i^{\text{th}}$ parameter of interest.
The diagonal elements of the positive-semidefinite QFIM matrix are the QFIs of each of the encoded parameters.

In order to obtain a scalar figure-of-merit in multiparameter estimation theory, which is much simpler to deal with than a matrix inequality, a cost or weight matrix $\mathcal{W}$ is introduced, which by definition is a real, positive-definite matrix. Using this weight matrix, a scalar multiparameter QCRB can be formulated as
\begin{equation*}
\label{cost}
\Tr{\left[\mathcal{W}\, V(\bm{x})  \right]} \geq \Tr{\left( \mathcal{W} \mathbb{\mathbf{F}}^{-1}_\mathbf{Q}(\rho(\bm{x}))\right)}.
\end{equation*}
\end{document}